\documentclass[12pt,amsymb,fullpage]{amsart}
\usepackage{amssymb,amscd,pstricks}

\newtheorem{theorem}{Theorem}[section]
\newtheorem{defn}[theorem]{Definition}

\newtheorem{lemma}[theorem]{Lemma}

\newtheorem{eple}[theorem]{Example}
\newtheorem{rmk}[theorem]{Remarks}
\newtheorem{dsc}[theorem]{Discussion}
\newtheorem{nota}[theorem]{Notation}

\newsavebox{\indbin}
\savebox{\indbin}{\begin{picture}(0,0)
\newlength{\gnu}
\settowidth{\gnu}{$\smile$} \setlength{\unitlength}{.5\gnu}
\put(-1,-.65){$\smile$} \put(-.25,.1){$|$}
\end{picture}}

\newcommand{\be}{\begin{enumerate}}
\newcommand{\bd}{\begin{defn}}
\newcommand{\bt}{\begin{theorem}}
\newcommand{\bl}{\begin{lemma}}
\newcommand{\ee}{\end{enumerate}}
\newcommand{\ed}{\end{defn}}
\newcommand{\et}{\end{theorem}}
\newcommand{\el}{\end{lemma}}

\begin{document}
\title{Functions Analytic at Infinity and Normality}
\author{Tristram de Piro}
\address{Flat 3, Redesdale House, 85 The Park, Cheltenham, GL50 2RP }
 \email{t.depiro@curvalinea.net}
\thanks{}
\begin{abstract}
This paper explores the notion of analytic at infinity and normality.
\end{abstract}
\maketitle
\begin{defn}
\label{decrease}
A smooth real function $f:\mathcal{R}\rightarrow\mathcal{R}$ is non oscillatory if it is eventually monotone, that is there exist $r\in\mathcal{R}_{>0}$ for which $f|_{(r,\infty)}$ is increasing or decreasing and similarly for $f|_{(-\infty,r)}$. We say that $f$ is of very moderate decrease if there exists a constant $C\in\mathcal{R}_{>0}$ for which $|f(x)|\leq {C\over |x|}$ for $|x|>1$. We say that $f$ is of moderate decrease if there exists a constant $C\in\mathcal{R}_{>0}$ for which $|f(x)|\leq {C\over |x|^{2}}$ for $|x|>1$.\\
\end{defn}
\begin{defn}
\label{analytic}
Given a smooth real function $f:{\mathcal{R}\setminus W}\rightarrow\mathcal{R}$, where $W$ is a bounded closed set of $\mathcal{R}$, we say that $f$ is analytic at infinity, if there exist $\{\epsilon_{1},\epsilon_{2}\}\subset\mathcal{R}_{>0}$ such that, for $0<x<\epsilon_{1}$;\\

$f({1\over x})=a(x)$\\

where $a(x)=\sum_{n=1}^{\infty}a_{n}x^{n}$, $a_{n}\in\mathcal{R}$, $n\geq 1$, is an absolutely convergent real power series on the interval $(0,\epsilon_{1})$, and for $-\epsilon_{2}<x<0$;\\

$f({1\over x})=b(x)$\\

where $b(x)=\sum_{n=1}^{\infty}b_{n}x^{n}$, $b_{n}\in\mathcal{R}$, $n\geq 1$, is an absolutely convergent real power series on the interval $(-\epsilon_{2},0)$.\\

\end{defn}

\begin{lemma}
\label{modulus}
The functions $f_{d}(x)={1\over |x-d|}$ where $d\in\mathcal{R}$ are analytic at infinity.\\

\end{lemma}

\begin{proof}
If $d>0$, we have that for $0<x<{1\over d}$, using the formula for a geometric progression;\\

$f_{d}({1\over x})={1\over |{1\over x}-d|}$\\

$={|x|\over |1-dx|}$\\

$={x\over 1-dx}$\\

$x\sum_{n=0}^{\infty}(dx)^{n}$\\

$=\sum_{n=0}^{\infty}d^{n}x^{n+1}$\\

$=\sum_{n=1}^{\infty}d^{n-1}x^{n}$\\

and, for $-{1\over d}<x<0$\\

$f_{d}({1\over x})={1\over |{1\over x}-d|}$\\

$={|x|\over |1-dx|}$\\

$=-{x\over 1-dx}$\\

$-x\sum_{n=0}^{\infty}(dx)^{n}$\\

$=-\sum_{n=0}^{\infty}d^{n}x^{n+1}$\\

$=-\sum_{n=1}^{\infty}d^{n-1}x^{n}$\\

The facts that $f_{d}$ are smooth, for $x\neq d$, and the cases $d=0$ and $d<0$ are left to the reader.\\

\end{proof}

\begin{lemma}
\label{modulus2}
If $f:{\mathcal{R}\setminus W}\rightarrow\mathcal{R}$ is analytic at infinity, then so is $f'$, moreover $f'$ is of moderate decrease.\\
\end{lemma}
\begin{proof}
As $f$ is analytic at infinity, we have that, for $0<x<\epsilon_{1}$;\\

$f({1\over x})=\sum_{n=1}a_{n}x^{n}$\\

so that, by the chain rule, and the fact that $\sum_{n=1}^{\infty}na_{n}x^{n-1}$ is absolutely convergent for $0<x<\epsilon_{1}$;\\

$-{f'({1\over x})\over x^{2}}=\sum_{n=1}^{\infty}na_{n}x^{n-1}$\\

so that, rearranging;\\

$f'({1\over x})=-x^{2}\sum_{n=1}^{\infty}na_{n}x^{n-1}$\\

$=-\sum_{n=1}^{\infty}na_{n}x^{n+1}$\\

$=-\sum_{n=2}^{\infty}(n-1)a_{n-1}x^{n}$ $(C)$\\

A similar calculation holds for $-\epsilon_{2}<x<0$, to show that $f$ is analytic at infinity. From $(C)$, using the proof of Lemma \ref{modulus3}, we can show that $f'$ is of moderate decrease. \\
\end{proof}

\begin{lemma}
\label{modulus3}
If $f$ is analytic at infinity, then it is of very moderate decrease and non-oscillatory.\\

\end{lemma}
\begin{proof}
By the definition, we have that, for $0<x<\epsilon_{1}$, that;\\

$f({1\over x})=x^{n}u(x)$\\

where, assuming $f$ is non vanishing on $({1\over \epsilon_{1}},\infty))$, $n\geq 1$ and $u(0)\neq 0$, and $|u(x)|\leq M$ on $(0,\epsilon_{1})$, so that;\\

$|f({1\over x})|\leq Mx^{n}$\\

$|f(x)\leq {M\over x^{n}}$\\

for $x>{1\over \epsilon_{1}}$. Similar considerations apply for $x<0$, so that;\\

$|f(x)|\leq {N\over |x|^{m}}$\\

for $x<-{1\over \epsilon_{2}}$, $m\geq 1$, $N\in\mathcal{R}_{>0}$. Without loss of generality, assuming that $max(\epsilon_{1},\epsilon_{2})<1$, taking $D=max(M,N)$, $r=max({1\over \epsilon_{1}},{1\over \epsilon_{2}})$, $p=min(m,n)$, we obtain that;\\

$|f(x)|\leq {D\over |x|^p}\leq {D\over |x|}$\\

for $|x|>r$. As $f$ is smooth, $|f|_{[-r,r]}|\leq M<{Mr\over |x|}$, so $f$ is of very moderate decrease, taking $C=max(D,Mr)$.\\

By Lemma \ref{modulus2}, we have that $f'$ is analytic at infinity. If $f'=0$, the result follows. Otherwise, we clam that the zero set of $f'$, $Zero(f')$ is contained in $[-s,s]$, for some $s\in\mathcal{R}_{>0}$, in which case the result again follows. To see this is the case, suppose there exists a sequence of zeros of $f'$, $\{s_{n}:n\in\mathcal{N}\}$, for which $|s_{n}|\rightarrow \infty$, Without loss of generality we may assume that $s_{n}>0$, so that the  absolutely convergent power series $a(x)$ has infinitely many zeros in the interval $(0,\epsilon)$, for any $0<\epsilon<\epsilon_{1}$. Writing $a(x)=x^{n}u(x)$, with $u(0)\neq 0$, by continuity we can assume that $u(x)\neq 0$, in the interval $(0,\epsilon_{3})$, where $0<\epsilon_{3}<\epsilon_{1}$. Then $x^{n}u(x)=0$ iff $x=0$ or $u(x)=0$, iff $x=0$ or $u(x)=0$, iff $x=0$ or $x\notin (0,\epsilon_{3})$, which is a contradiction.\\

\end{proof}
\begin{rmk}
\label{paper}
The class of non-oscillatory functions, with very moderate decrease, was considered in the paper \cite{dep}, where we proved a Fourier inversion theorem.\\

\end{rmk}

\begin{defn}
\label{normaldef}
Let $f:{\mathcal{R}^{2}\setminus W}\rightarrow \mathcal{R}$ be smooth, with $W$ closed and bounded, we say that $f$ is of very moderate decrease if there exists a constant $C\in\mathcal{R}_{>0}$ such that $|f(x,y)|\leq {C\over |(x,y)|}$, for $|(x,y)|>1$. We say that $f$ is of moderate decrease if there exists a constant $C\in\mathcal{R}_{>0}$ such that $|f(x,y)|\leq {C\over |(x,y)|^{2}}$, for $|(x,y)|>1$. We say that $f$ is of moderate decrease $n$, if there exists a constant $C\in\mathcal{R}_{>0}$ such that $|f(x,y)|\leq {C\over |(x,y)|^{n}}$, for $|(x,y)|>1$, with $n\geq 2$. We say that $f$ is normal, if;\\

$(i)$. For $x\in\mathcal{R}$, $f_{x}(y)$ is analytic at infinity.\\

$(ii)$. For $y\in\mathcal{R}$, $f_{y}(x)$ is analytic at infinity.\\

$(iii)$. $f$ is of very moderate decrease.\\

$(iv)$. The higher derivatives ${\partial f\over \partial x}$ and ${\partial f\over \partial y}$, are of moderate decrease.\\

$(v)$. There is a uniform bound $val$ in the number of zeros of $f_{x},(f_{x})'(f_{x})''$ and $f_{y},(f_{y})',(f_{y})''$.\\

We say that $f$ is quasi normal if $(i)-(iv)$ hold and $(v)'$, for sufficiently large $x$, the zeros of $f_{x}$ are contained in the union of bounded intervals $(-M_{x},-N_{x})\cup (N_{x},M_{x})$, with $M_{x}-N_{x}$ uniformly bounded in $x$, and similarly for $\{(f_{x})'(f_{x})''\}$, with corresponding $M'_{x},M''_{x},N'_{x},N''_{x}$, and $\{f_{y},(f_{y})'(f_{y})''\}$, with corresponding $\{M_{y}, M'_{y},M''_{y},N_{y},N'_{y},N''_{y}\}$. We say that $f$ is quasi split normal if $(i)-(iv)$ hold and $(v)''$, for sufficiently large $(x,y)$, $f(x,y)=f_{1}(x,y)+f_{2}(x,y)$, with $f_{1}$ and $f_{2}$ quasi normal, and $\{f,f_{1},f_{2}\}$ are smooth.\\

\end{defn}
\begin{lemma}
\label{example}
If $\{a,b\}\subset\mathcal{R}$, and $W$ is a closed ball $B((a,b),s)$, $s>0$, containing $(a,b)$, the function $f(x,y)|_{{\mathcal{R}^{2}\setminus W}}$\\

where $f(x,y)={1\over |(x,y)-(a,b)|}$, $(x,y)\neq (a,b)$\\

is normal.\\

If $\rho\geq 0$ is continuous with compact support, $\rho\neq 0$, the function;\\

$f(x,y)=\int_{\mathcal{R}^{2}}{\rho(x',y')\over |(x,y)-(x',y')|}dx'dy'$\\

is quasi normal, and if $\rho$ is smooth with compact support, $\rho\neq 0$;\\

$f(x,y)=\int_{\mathcal{R}^{2}}{\rho(x',y')\over |(x,y)-(x',y')|}dx'dy'$\\

is quasi split normal.\\

\end{lemma}
\begin{proof}
Fix $x_{0}\in\mathcal{R}$, then;\\

$f_{x_{0}}(y)={1\over ((x_{0}-a)^{2}+(y-b)^{2})^{1\over 2}}$\\

Without loss of generality, assuming that $x_{0}\neq a$, we have that, for $y>0$;\\

 $f_{x_{0}}({1\over y})={1\over ((x_{0}-a)^{2}+({1\over y}-b)^{2})^{1\over 2}}$\\

 $={y\over (y^{2}(x_{0}-a)^{2}+1-2yb+y^{2}b^{2})^{1\over 2}}$\\

For $y<1$, we have that $y^{2}<y$, so that;\\

$|y^{2}(x_{0}-a)^{2}-2yb+y^{2}b^{2}|<y|(x_{0}-a)^{2}+2|b|+b^{2}|<1$\\

iff $y<{1\over |(x_{0}-a)^{2}+2|b|+b^{2}|}$\\

and, applying Newton's theorem, with $b_{n}={(-1)^{n}(2n)!\over 2^{2n}(n!)^{2}}$;\\

$f_{x_{0}}({1\over y})=y\sum_{n=0}^{\infty}b_{n}(y^{2}(x_{0}-a)^{2}-2yb+y^{2}b^{2})^{n}$\\

$=\sum_{n=1}^{\infty}a_{n}y^{n}$\\

is an absolutely convergent power series in $y$ of order $1$. A similar result holds for $y<0$, so that $f_{x_{0}}$ is analytic at infinity. Similarly, fixing $y_{0}\in\mathcal{R}$, we can see that $f_{y_{0}}(x)$ is analytic at infinity. Let $\overline{r}=(x,y)$, $r=|x,y|$, $\overline{a}=(a,b)$, then, switching to polars, for $0\leq \theta<2\pi$;\\

$lim_{r\rightarrow\infty} r f(r,\theta)=lim_{r\rightarrow \infty}{r\over |\overline{r}(r,\theta)-\overline{a}|}$\\

$=lim_{r\rightarrow\infty}{1\over |{\overline{r}(r,\theta)\over r}-{\overline{a}\over r}|}$\\

$=lim_{r\rightarrow\infty}{1\over |\hat{\overline{r}}(r,\theta)|}$\\

$=1$\\

so that fixing a closed ball $B(\overline{0},s)\supset W$, using the fact that $f$ is smooth on $B(\overline{0},s)^{c}$, $r|f|\leq D$, where $D\in\mathcal{R}_{>0}$, on $B(\overline{0},s)^{c}$, so that $|f|\leq {D\over r}$, for $r>s$. As $f$ is continuous on ${B(\overline{0},s)\setminus W^{\circ}}$, it is bounded, by compactness of ${B(\overline{0},s)\setminus W^{\circ}}$, so that $|f|\leq M$ for $|x|\geq 1$. It follows that $|f|\leq {C\over r}$, for $|x|\geq 1$, where $C=max(D,Ms)$. Therefore, $f|_{{\mathcal{R}^{2}\setminus W}}$ is of very moderate decrease.\\

We have that $f_{x}$ has no zeros, similarly for $f_{y}$, and, by the chain rule;\\

${\partial f\over \partial x}=-{1\over 2}2(x-a){1\over ((x-a)^{2}+(y-b)^{2})^{3\over 2}}$\\

$=-{x-a\over ((x-a)^{2}+(y-b)^{2})^{3\over 2}}$\\

so that $(f_{y})'$ has a zero when $x=a$, similarly $(f_{x})'$ has a zero when $y=b$. We have;\\

$lim_{r\rightarrow\infty}r^{2}|{\partial f\over \partial x}|=lim_{r\rightarrow\infty}{(x-a)r^{2}\over |\overline{r}-\overline{a}|^{3}}$\\

$\leq lim_{r\rightarrow\infty}{r^{2}|\overline{r}-\overline{a}|\over |\overline{r}-\overline{a}|^{3}}$\\

$=lim_{r\rightarrow\infty}{r^{2}\over |\overline{r}-\overline{a}|^{2}}$\\

$=lim_{r\rightarrow\infty}{1\over |{\overline{r}\over r}-{\overline{a}\over r}|^{2}}$\\

$=lim_{r\rightarrow\infty}{1\over |\hat{\overline{r}}|^{2}}$\\

$=1$\\

so that fixing a closed ball $B(\overline{0},s)\supset W$, using the fact that ${\partial f\over \partial x}$ is smooth on $B(\overline{0},s)^{c}$, $r^{2}|{\partial f\over \partial x}|\leq D$, where $D\in\mathcal{R}_{>0}$, on $B(\overline{0},s)^{c}$, so that $|{\partial f\over \partial x}|\leq {D\over r^{2}}$, for $r>s$. As ${\partial f\over \partial x}$ is continuous on ${B(\overline{0},s)\setminus W^{\circ}}$, it is bounded, by compactness of ${B(\overline{0},s)\setminus W^{\circ}}$ again, so that $|{\partial f\over \partial x}|\leq M$ for $|x|\geq 1$. It follows that $|{\partial f\over \partial x}|\leq {C\over |\overline{x}|^{2}}$, for $|\overline{x}|\geq 1$, where $C=max(D,Ms^{2})$. A similar proof works for ${\partial f\over \partial y}$, and the higher derivatives ${\partial f^{n+m}\over \partial x^{m}\partial y^{n}}$, $n+m\geq 1$, the details are left to the reader, so that $\{{\partial^{n+m} f\over \partial x^{n}\partial y^{n}}_{{\mathcal{R}^{2}\setminus W}}:n+m\geq 1\}$ are of moderate decrease. We have that, by the product rule;\\

${\partial^{2} f\over \partial x^{2}}=-{1\over ((x-a)^{2}+(y-b)^{2})^{3\over 2}}+{3(x-a)^{2}\over ((x-a)^{2}+(y-b)^{2})^{5\over 2}}=0$\\

iff $3(x-a)^{2}-[(x-a)^{2}+(y-b)^{2}]=0$\\

iff $2(x-a)^{2}=(y-b)^{2}$\\

iff $x-a={1\over \sqrt{2}}(y-b)$ or $x-a=-{1\over \sqrt{2}}(y-b)$\\

so that $(f_{y})''$ has at most $2$ zeros for $y\in\mathcal{R}$ and we can take $val=2$. A similar result holds for $(f_{x})''$.\\

For the second claim, note that if $(x,y)\in Supp(\rho)$, then;\\

$f(x,y)=\int_{\mathcal{R}^{2}}{\rho(x',y')\over |(x,y)-(x',y')|}dx'dy'$\\

$=\int_{\mathcal{R}^{2}}{\rho(x-x',y-y')\over |(x',y')|}dx'dy'$\\

$=\int_{0,2\pi}\int_{\mathcal{R}_{>0}}{\rho_{x,y}(r,\theta)\over r}rsin(\theta)drd\theta$\\

$=\int_{0,2\pi}\int_{\mathcal{R}_{>0}}\rho_{x,y}(r,\theta)sin(\theta)drd\theta$\\

so that;\\

$|f(x,y)|\leq \int_{0,2\pi}\int_{\mathcal{R}_{>0}}|\rho_{x,y}(r,\theta)|drd\theta$\\

$\leq 2 M\pi R(x,y)$\\

where $Supp_{x,y}(\rho)\subset B(\overline{0},R(x,y))$, $|\rho|\leq M$, so that $f$ is defined everywhere. If $\rho$ is smooth, we have that $f$ is smooth, as;\\

${\partial^{i+j}f\over \partial x^{i}\partial y^{j}}(x,y)={\partial^{i+j} \int_{\mathcal{R}^{2}}{\rho(x-x',y-y')\over |(x',y')|}dx'dy'\over \partial x^{i}\partial y^{j}}$\\

$=\int_{\mathcal{R}^{2}}{{\partial^{i+j}\rho\over \partial x^{i}\partial y^{j}}(x-x',y-y')\over |(x',y')|}dx'dy'$\\

with ${\partial^{i+j}\rho\over \partial x^{i}\partial y^{j}}$ having compact support again. If $\rho$ is continuous, but not necessarily smooth, we have that, for $(x,y)\notin Supp(\rho)$;\\

${\partial^{i+j}f\over \partial x^{i}\partial y^{j}}(x,y)=\int_{\mathcal{R}^{2}}\rho(x',y'){\partial^{i+j}\over \partial x^{i}\partial y^{j}}({1\over |(x,y)-(x',y')|})dx'dy'$\\

so that $f$ is smooth, outside $Supp(\rho)$. For $(i)$, in Definition \ref{normaldef}, we have, for $x_{0}\in\mathcal{R}$, $y>0$;\\

$f_{x_{0}}({1\over y})=\int_{\mathcal{R}^{2}}{\rho(x',y')\over |(x_{0},{1\over y})-(x',y')|}dx'dy'$\\

$=y\int_{\mathcal{R}^{2}}{\rho(x',y')\over (1+y^{2}(x_{0}-x')^{2}-2yy'+y^{2}y'^{2})^{1\over 2}}$\\

so that, with $y<1$, $y^{2}<y$, letting;\\

$M_{x_{0}}=sup_{(x',y')\in Supp(\rho)}|(x_{0}-x'|^{2}+2|y'|+y'^{2}|$, if $y<{1\over M_{x_{0}}}$, then;\\

$|y^{2}(x_{0}-x')^{2}-2yy'+y^{2}y'^{2}|<y|(x_{0}-x')^{2}+2|y'|+y'^{2}|<1$\\

 so that, we can apply Newton's theorem uniformly in $(x',y')\in Supp(\rho)$, to obtain;\\

$f_{x_{0}}({1\over y})=y\int_{Supp(\rho)}\rho(x',y')(\sum_{n=0}^{\infty}b_{n}(y^{2}(x_{0}-x')^{2}-2yy'+y^{2}y'^{2})^{n})dx'dy'$\\

where $b_{n}$ is as above. With $y<1$ again, $|x'|\leq M$, $|y'|\leq M$, for $(x',y')\subset Supp(\rho)$, $|\rho|\leq N$, $y<{1\over (|x_{0}|+M)^{2}+2M+M^{2}}$, we have, applying the DCT;\\

$|f_{x_{0}}({1\over y})|\leq yN(2M)^{2}\sum_{n=0}^{\infty}|b_{n}|((|x_{0}|+M)^{2}y+2yM+yM^{2})^{n}$\\

$\leq 4yNM^{2}\sum_{n=0}^{\infty}|b_{n}|y^{n}((|x_{0}|+M)^{2}+2M+M^{2})^{n}$\\

$\leq 4yNM^{2}\sum_{n=0}^{\infty}y^{n}((|x_{0}|+M)^{2}+2M+M^{2})^{n}$\\

defines an absolutely convergent series. A similar proof works for $y<0$. $(ii)$ is similar. For $(iii)$, with $Supp(\rho)\subset B(\overline{0},M)$, $|\rho|\leq N$, if $|(x,y)|\geq 2M$, and $(x',y')\in Supp(\rho)$, $|(x,y)-(x',y')|\geq {|(x,y)\over 2}$, so that ${1\over |(x,y)-(x',y')|}\leq {2\over |(x,y)|}$, and;\\

$|f(x,y)|= |\int_{Supp(\rho)}{\rho(x',y')\over |(x,y)-(x',y')|}dx'dy'|$\\

$\leq {2\over |(x,y)|}\int_{Supp(\rho)}|\rho(x',y')|dx'dy'$\\

$\leq {2\pi M^{2}N\over |(x,y)|}$\\

For $(iv)$, we have that, if $|(x,y)|\geq max(2M,1)$, and $(x',y')\in Supp(\rho)$, ${1\over |(x,y)-(x',y')|^{3}}\leq {8\over |(x,y)|^{3}}$, $|x-x'|\leq |(x,y)-(x',y')|\leq |(x,y)|+M$;\\

$|{\partial f\over \partial x}|=|\int_{Supp(\rho)}{-\rho(x',y')(x-x')\over |(x,y)-(x',y')|^{3}}dx'dy'|$\\

$\leq {8(|(x,y)|+M)\over |(x,y)|^{3}}\int_{Supp(\rho)}|\rho(x',y')|dx'dy'$\\

$\leq {8\pi M^{2}N(|(x,y)|+M)\over |(x,y)|^{3}}$\\

$\leq {8\pi (M^{2}+M)N \over |(x,y)|^{2}}$\\

The proof for ${\partial f\over \partial y}$ is similar. For $(v)'$, we have that if $\rho\geq 0$ is continuous with compact support, $\rho\neq 0$, that $f>0$. For $(x,y)\notin Supp(\rho)$, we have that;\\

${\partial f\over \partial x}=\int_{\mathcal{R}^{2}}-{\rho(x',y')(x-x')\over |(x,y)-(x',y')|^{3}}dx'dy'$\\

and for $x>M$, $(x-x')>0$, for $x<M$, $(x-x')<0$, where $x'\in Supp(\rho)$, so that for $|y|>M$, $(f_{y})'<0$, for $x>M$, and $(f_{y})'>0$, for $x<M$, as $\rho\geq 0$. In particular, the zeros of $f_{y}$ are contained in the interval $(-(M+1),(M+1))$, with the length $2(M+1)$ of the interval, uniformly bounded in $y$. We have that;\\

${\partial^{2} f\over \partial x^{2}}=\int_{\mathcal{R}^{2}}\rho(x',y')[-{1\over |(x,y)-(x',y')|^{3}}+{3(x-x')^{2}\over |(x,y)-(x',y')|^{5}}]dx'dy'$\\

$=\int_{\mathcal{R}^{2}}\rho(x',y')[{2(x-x')^{2}-(y-y')^{2}\over |(x,y)-(x',y')|^{5}}]dx'dy'$\\

so that if $2(x-x')^{2}-(y-y')^{2}>0$, for $(x'y')\in Supp(\rho)$, $(f_{y})''>0$, and if $2(x-x')^{2}-(y-y')^{2}<0$, for $(x'y')\in Supp(\rho)$, $(f_{y})''<0$. We have that;\\

$2(x-x')^{2}-(y-y')^{2}>0$\\

iff $|x-x'|>{1\over \sqrt{2}}|y-y'|$ $(i)$\\

$2(x-x')^{2}-(y-y')^{2}<0$\\

iff $|x-x'|<{1\over \sqrt{2}}|y-y'|$ $(ii)$\\

so that if $|x|>M+{1\over \sqrt{2}}(|y|+M)$, $(i)$ holds, and if $|x|<-M+{1\over \sqrt{2}}(|y|-M)$, $(ii)$ holds, for $|y|> 2M$. In particularly, the zeros of $(f_{y})''$ are contained in the intervals $(-M+{1\over \sqrt{2}}(|y|-M),M+{1\over \sqrt{2}}(|y|+M))\cup (-M-{1\over \sqrt{2}}(|y|+M),M-{1\over \sqrt{2}}(|y|-M))$, with the length of the intervals, $(2+\sqrt{2})M$ uniform in $|y|>2M$. For $(v)''$, we can split $\rho$ into $\rho^{+}$ and $\rho^{-}$ which are continuous with compact support, and use the previous result, noting that quasi normal implies quasi split normal, as $f={f\over 2}+{f\over 2}$.\\

\end{proof}

\begin{rmk}
\label{uniformbounds}
The above functions have stronger properties, for example it can probably be shown that ${\partial f^{n+m}\over \partial x^{m}\partial y^{n}}$ has moderate decrease $n+m+1$, for $n+m\geq 1$, and there is a uniform bound on the zeros of all the higher derivatives $f_{x}^{(n)}$ and $f_{y}^{(m)}$, for all $n\geq 1$, $m\geq 1$. The details are left as an exercise. We include the case of $(f_{y})'''$ for the first function;\\

${\partial^{3} f\over \partial x^{3}}={9(x-a)\over ((x-a)^{2}+(y-b)^{2})^{5\over 2}}-{15(x-a)^{3}\over ((x-a)^{2}+(y-b)^{2})^{7\over 2}}=0$\\

iff $9(x-a)[(x-a)^{2}+(y-b)^{2}]-15(x-a)^{3}=0$\\

iff $9[(x-a)^{2}+(y-b)^{2}]-15(x-a)^{2}=0$\\

iff $6(x-a)^{2}=9(y-b)^{2}$\\

iff $x-a=\sqrt{3\over 2}(y-b)$ or $x=-\sqrt{3\over 2}(y-b)$\\

so that $(f_{y})'''$ has at most $2$ zeros for $y\in\mathcal{R}$ and we can take $val=2$. A similar result holds for $(f_{x})'''$.\\
\end{rmk}

\begin{lemma}
\label{normal}
Let $f:\mathcal{R}^{2}\rightarrow\mathcal{R}$ be smooth and normal, then, for $\{x,y\}\subset\mathcal{R}$, $k_{1}\neq 0$, $k_{2}\neq 0$;\\

$F(k_{1},y)=lim_{r\rightarrow\infty}\int_{-r}^{r}f(x,y)e^{-ik_{1}x}dx$\\

$G(x,k_{2})=lim_{r\rightarrow\infty}\int_{-r}^{r}f(x,y)e^{-ik_{2}y}dy$\\

both exist and $F(k_{1},y)$, $G(x,k_{2})$ are of moderate decrease.\\

\end{lemma}
\begin{proof}
The first claim follows from \cite{dep} together with Lemma \ref{modulus3} and $(i),(ii)$ in Definition \ref{normaldef}. In fact, the first integral is indefinite, in the sense that we could define it as;\\

$F(k_{1},y)=lim_{r\rightarrow\infty, s\rightarrow\infty}(\int_{a}^{r}f(x,y)e^{-ik_{1}x}dx+\int_{-s}^{a}f(x,y)e^{-ik_{1}x}dx)$\\

for a choice of $a\in\mathcal{R}$, and similarly for $G(x,k_{2})$. We then have, using integration by parts, for $y_{0}\in\mathcal{R}$;\\

$\int_{-r}^{r}f(x,y_{0})e^{-ik_{1}x}dx=[{if(x,y_{0})e^{-ik_{1}x}\over k_{1}}]_{-r}^{r}-\int_{-r}^{r}{i{\partial f\over \partial x}(x,y_{0})e^{-ik_{1}x}\over k_{1}}dx$\\

so that, as $f_{y_{0}}$ is of very moderate decrease, by $(iii)$, integrating by parts;\\

$F(k_{1},y_{0})=lim_{r\rightarrow\infty}\int_{-r}^{r}f(x,y_{0})e^{-ik_{1}x}dx$\\

$=lim_{r\rightarrow\infty}([{if(x,y_{0})e^{-ik_{1}x}\over k_{1}}]_{-r}^{r}-\int_{-r}^{r}{i{\partial f\over \partial x}(x,y_{0})e^{-ik_{1}x}\over k_{1}}dx)$\\

$=lim_{r\rightarrow\infty}(-\int_{-r}^{r}{i{\partial f\over \partial x}(x,y_{0})e^{-ik_{1}x}\over k_{1}}dx)$\\

$=-{i\over k_{1}}\int_{-\infty}^{\infty}{\partial f\over \partial x}(x,y_{0})e^{-ik_{1}x}dx$\\

the last integral being definite, as $({\partial f\over \partial x})_{y_{0}}$ is of moderate decrease, using $(iv)$. It follows that $F(k_{1},y)$ is smooth, as differentiating under the integral sign is justified by the DCT, MVT and $(iv)$, with;\\

$F^{(n)}(k_{1},y)=-{i\over k_{1}}\int_{-\infty}^{\infty}{\partial^{1+n}f\over \partial x\partial y^{n}}(x,y)e^{-ik_{1}x}dx$\\

We have that $({\partial f\over \partial x})_{y_{0}}$ is analytic at infinity by $(ii)$ and Lemma \ref{modulus2}. By $(v)$ in the definition of normality, we can assume that for sufficiently large $y$, $({\partial f\over \partial x})_{y}$ is monotone and positive/negative in the intervals $(-\infty,a_{1}(y)),\ldots,(a_{val}(y),\infty)$, where $a_{1},\ldots,a_{val}$ vary continuously with $y$. Splitting the integral into $cos(k_{1}x)$ and $sin(k_{1}x)$ components, in a similar calculation to \cite{dep2}, for any interval of length at least ${2\pi\over |k_{1}|}$, we obtain an alternating cancellation in the contribution to the integral of at most ${2\pi ||f_{y}'||\over |k_{1}|}$ and for any interval of length at most ${2\pi\over |k_{1}|}$, we obtain a contribution to the integral of at most ${4\pi ||f_{y}'||\over |k_{1}|}$. It follows that, for sufficiently large $y$, using the fact that ${\partial f\over \partial x}$ has moderate decrease;\\

$|-{i\over k_{1}}\int_{-\infty}^{\infty}{\partial f\over \partial x}(x,y)e^{-ik_{1}x}dx|$\\

$\leq {(val+1)\over |k_{1}|}{4\pi ||f_{y}'||\over |k_{1}|}$\\

$={4\pi(val+1)||f_{y}'||\over |k_{1}|^{2}}$\\

$\leq {4\pi(val+1)C\over |k_{1}|^{2}|y|^{2}}$ $(*)$\\

so that;\\

$|F(k_{1},y)|\leq {D\over |y|^{2}}$\\

for sufficiently large $y$, with $D={4\pi(val+1)C\over |k_{1}|^{2}}$. As $F(k_{1},y)$ is smooth, we obtain that $F(k_{1},y)$ is of moderate decrease. Similarly, we can show that $G(x,k_{2})$ is of moderate decrease, for $k_{2}\neq 0$.
\end{proof}
\begin{lemma}
\label{quasinormal}
The same result as Lemma \ref{normal} holds, with the assumption that $f$ is smooth and quasi normal or smooth and quasi split normal.

\end{lemma}
\begin{proof}
We just have to replace the uses of $(v)$ in Definition \ref{normaldef} with the use of $(v)'$ or $(v)''$. For $(v)'$, we have that, for sufficiently large $y$, $({\partial f\over \partial x})_{y}$ is monotone and positive/negative outside a finite union $I_{y}$ of $S$ intervals whose total length is uniformly bounded by a constant $R\in\mathcal{R}_{>0}$ which is independent of $y$. By the usual argument, we obtain a contribution of at most $||f_{y}'||R$ to the integral over this interval. As before, splitting the integral into $cos(k_{1}x)$ and $sin(k_{1}x)$ components, we can bound the contribution of the remaining integral by ${(S+1)4\pi||f_{y}'||\over |k_{1}|}$ to give a total bound for the calculation in $(*)$ of Lemma \ref{normal} of $({(S+1)4\pi\over |k_{1}|}+R)||f_{y}'||$, and we can then, as before, use the fact that ${\partial f\over \partial x}$ has moderate decrease. For $(v)''$, we can split $f_{y}'=f_{1,y}'+f_{2,y}'$, and repeating the above argument twice obtain a total bound of $({(S+1)4\pi\over |k_{1}|}+R)(||f_{1,y}'||+||f_{2,y}'||)$, where the bounds $S$ and $R$ work for both $f_{1,y}'$ and $f_{2,y}'$. We can then use the the fact that ${\partial f_{1}\over \partial x}$ and ${\partial f_{2}\over \partial x}$ have moderate decrease.\\
\end{proof}
\begin{lemma}
\label{normal1}
Let hypotheses and notation be as in Lemma \ref{normal}, then we can define, for $k_{1}\neq 0$, $k_{2}\neq 0$;\\

$F(k_{1},k_{2})=\int_{-\infty}^{\infty}F(k_{1},y)e^{-ik_{2}y}dy$\\

$G(k_{1},k_{2})=\int_{-\infty}^{\infty}G(x,k_{2})e^{-ik_{1}x}dx$\\

We have that;\\

$F(k_{1},k_{2})=G(k_{1},k_{2})=lim_{m\rightarrow\infty,n\rightarrow\infty}\int_{-m}^{m}\int_{-n}^{n}f(x,y)e^{-ik_{1}x}e^{-ik_{2}y}dxdy$\\
\end{lemma}

\begin{proof}

The definition follows from the previous lemma, as $F(k_{1},y)$ and $G(x,k_{2})$ are smooth and of moderate decrease. Integrating by parts again, using the fact that $f$ is of very moderate decrease, we have that, for $n\in\mathcal{R}_{>0}$;\\

$\int_{|x|>n}f(x,y)e^{-ik_{1}x}dx$\\

$=([{if(x,y)e^{-ik_{1}x}\over k_{1}}]_{n}^{\infty}+[{if(x,y)e^{-ik_{1}x}\over k_{1}}]_{-\infty}^{-n}-\int_{|x|>n}{i{\partial f\over \partial x}(x,y)e^{-ik_{1}x}\over k_{1}}dx)$\\

$=(-{if(n,y)e^{-ik_{1}n}\over k_{1}}+{if(-n,y)e^{ik_{1}n}\over k_{1}}-\int_{|x|>n}{i{\partial f\over \partial x}(x,y)e^{-ik_{1}x}\over k_{1}}dx)$\\

We have that, using the above calculation $(*)$ in Lemma \ref{normal}, $f$ of very moderate decrease;\\

$|\int_{|y|\leq m}({if(-n,y)e^{ik_{1}n}\over k_{1}}-{if(n,y)e^{-ik_{1}n}\over k_{1}})e^{-ik_{2}y}dy|$\\

$\leq {1\over k_{1}}|\int_{|y|\leq m}f(-n,y)e^{-ik_{2}y}dy|+{1\over k_{1}}|\int_{|y|\leq m}f(n,y)e^{-ik_{2}y}dy|$\\

$\leq {4\pi(val+1)||f_{-n}|_{|y|\leq m}||\over |k_{1}|^{2}}+{4\pi(val+1)||f_{n}|_{|y|\leq m}||\over |k_{1}|^{2}}$\\

$\leq {8\pi(val+1)C\over |k_{1}|^{2}n}$\\

and using the previous calculation $(*)$ again in Lemma \ref{normal}, we have that;\\

$|-{i\over k_{1}}\int_{|x|>n}{\partial f\over \partial x}(x,y)e^{-ik_{1}x}dx|$\\

$\leq {(val+1)\over |k_{1}|}{4\pi ||f_{y}'||_{|x|>n}\over |k_{1}|}$\\

$={4\pi(val+1)||f_{y}'||_{|x|>n}\over |k_{1}|^{2}}$\\

$\leq {4\pi(val+1)C\over |k_{1}|^{2}(y^{2}+n^{2})}$\\

so that;\\

$|-{i\over k_{1}}\int_{|y|\leq m}\int_{|x|>n}{\partial f\over \partial x}(x,y)e^{-ik_{1}x}dxe^{-ik_{2}y}dy|$\\

$\leq \int_{|y\leq m}{4\pi(val+1)C\over |k_{1}|^{2}(y^{2}+n^{2})}dy$\\

$\leq \int_{-\infty}^{\infty}{4\pi(val+1)C\over |k_{1}|^{2}(y^{2}+n^{2})}dy$\\

$=\int_{-\infty}^{\infty}{1\over n^{2}}{4\pi(val+1)C\over |k_{1}|^{2}(1+{y^{2}\over n^{2}})}dy$\\

$\leq {1\over n^{2}}{4\pi(val+1)C\over |k_{1}|^{2}}ntan^{-1}({y\over n})|^{\infty}_{-\infty}$\\

$\leq {\pi\over n}{4\pi(val+1)C\over |k_{1}|^{2}}$\\

The above calculations combine, to give that;\\

$|\int_{|y|\leq m}lim_{n\rightarrow\infty}\int_{-n}^{n}f(x,y)e^{-ik_{1}x}e^{-ik_{2}y}dxdy-\int_{|y|\leq m}\int_{|x|\leq n}f(x,y)e^{-ik_{1}x}e^{-ik_{2}y}dxdy|$\\

$=|\int_{|y|\leq m}\int_{|x|>n}f(x,y)e^{-ik_{1}x}e^{-ik_{2}y}dxdy|$\\

$\leq {8\pi(val+1)C\over |k_{1}|^{2}n}+{\pi\over n}{4\pi(val+1)C\over |k_{1}|^{2}}$\\

so that $lim_{n\rightarrow \infty}s_{m,n}=s_{m}$, uniformly in $m$, where;\\

$s_{n,m}=\int_{|y|\leq m}\int_{|x|\leq n}f(x,y)e^{-ik_{1}x}e^{-ik_{2}y}dxdy$\\

$s_{m}=\int_{|y|\leq m}F(k_{1},y)e^{-ik_{2}y}dy$\\

By the Moore-Osgood Theorem, we obtain the result.\\

\end{proof}
\begin{lemma}
\label{quasinormal1}
The same result as Lemma \ref{normal1} holds with the assumption that $f$ is smooth and quasi normal or smooth and quasi split normal.\\
\end{lemma}
\begin{proof}
Again, we can replace the estimates from $(v)$ in Definition \ref{normaldef}, used in the proof of Lemma \ref{normal1}, with the estimate used in Lemma \ref{quasinormal}.
\end{proof}
\begin{defn}
\label{normal3}
Let $f:{\mathcal{R}^{3}\setminus W}\rightarrow \mathcal{R}$ be smooth, with $W$ closed and bounded, we say that $f$ is of very moderate decrease if there exists a constant $C\in\mathcal{R}_{>0}$ such that $|f(x,y,z)|\leq {C\over |(x,y,z)|}$, for $|(x,y,z)|>1$. We say that $f$ is of moderate decrease if there exists a constant $C\in\mathcal{R}_{>0}$ such that $|f(x,y,z)|\leq {C\over |(x,y,z)|^{2}}$, for $|(x,y,z)|>1$. We say that $f$ is of moderate decrease $n$, if there exists a constant $C\in\mathcal{R}_{>0}$ such that $|f(x,y,z)|\leq {C\over |(x,y,z)|^{n}}$, for $|(x,y,z)|>1$, with $n\geq 2$. We say that $f$ is normal, if;\\

$(i)$. For $x\in\mathcal{R}$, $f_{x}(y,z)$ is normal.\\

$(ii)$. For $y\in\mathcal{R}$, $f_{y}(x,z)$ is normal.\\

$(iii)$. For $z\in\mathcal{R}$, $f_{z}(x,y)$ is normal.\\

$(iv)$. $f$ is of very moderate decrease.\\

$(v)$. The higher derivatives ${\partial^{i+j+k} f\over \partial x^{i}\partial y^{j}\partial z^{k}}$ are of moderate decrease $i+j+k+1$, for $i+j+k\geq 1$.\\

$(vi)$. There is a uniform bound $val(x,y)$ in the number of zeros of;\\

$f_{x,y},(f_{x,y})',(f_{x,y})'',(f_{x,y})''',(f_{x,y})''''$\\

and similarly, for $f_{x,z},f_{y,z}$.\\

We say that $f$ is quasi normal, if;\\

$(i)'$. For $x\in\mathcal{R}$, $f_{x}(y,z)$ is quasi normal, and, similarly for $(ii),(iii)$, $(iv),(v)$ hold and;\\

$(vi)'$. For sufficiently large $(x,y)$, the zeros of;\\

$f_{x,y},(f_{x,y})',(f_{x,y})'',(f_{x,y})''',(f_{x,y})''''$\\

are contained in a finite union of $S$ intervals, with total length $R$, uniform in $(x,y)$.\\

and similarly, for $f_{x,z},f_{y,z}$.\\

We say that $f$ is quasi split normal, if;\\

$(i)''$. For $x\in\mathcal{R}$, $f_{x}(y,z)$ is quasi split normal, and, similarly for $(ii),(iii)$, $(iv),(v)$ hold and;\\

$(vi)''$. For sufficiently large $(x,y)$, $f=f_{1}+f_{2}$, with $f,f_{1},f_{2}$ smooth and with $f_{1},f_{2}$ having the property $(vi)'$.\\
\end{defn}
\begin{lemma}
\label{examples2}
If $\{a,b,c\}\subset\mathcal{R}$, and $W$ is a closed ball $B((a,b,c),s)$, $s>0$, containing $(a,b,c)$, the function $f(x,y,z)|_{{\mathcal{R}^{3}\setminus W}}$\\

where $f(x,y,z)={1\over |(x,y,z)-(a,b,c)|}$, $(x,y,z)\neq (a,b,c)$\\

is normal.\\

If $\rho\geq 0$ is continuous with compact support, $\rho\neq 0$, the function;\\

$f(x,y,z)=\int_{\mathcal{R}^{3}}{\rho(x',y',z')\over |(x,y,z)-(x',y',z')|}dx'dy'dz'$\\

is quasi normal.\\

If $\rho$ is smooth with compact support, $\rho\neq 0$, the function;\\

$f(x,y,z)=\int_{\mathcal{R}^{3}}{\rho(x',y',z')\over |(x,y,z)-(x',y',z')|}dx'dy'dz'$\\

is quasi split normal.\\
\end{lemma}
\begin{proof}
For the first claim, we have to show, for $x_{0}\in\mathcal{R}$, that $f_{x_{0}}(y,z)$ is normal. Fix $z_{0}\in\mathcal{R}$, then;\\

$f_{x_{0},z_{0}}(y)={1\over ((x_{0}-a)^{2}+(z_{0}-c)^{2}+(y-b)^{2})^{1\over 2}}$\\

Without loss of generality, assuming that $x_{0}\neq a, z_{0}\neq c$, we have that, for $y>0$;\\

 $f_{x_{0},z_{0}}({1\over y})={1\over ((x_{0}-a)^{2}+(z_{0}-c)^{2}+({1\over y}-b)^{2})^{1\over 2}}$\\

 $={y\over (y^{2}(x_{0}-a)^{2}+y^{2}(z_{0}-c)^{2}+1-2yb+y^{2}b^{2})^{1\over 2}}$\\

For $y<1$, we have that $y^{2}<y$, so that;\\

$|y^{2}(x_{0}-a)^{2}+y^{2}(z_{0}-c)^{2}-2yb+y^{2}b^{2}|<y|(x_{0}-a)^{2}+(z_{0}-c)^{2}+2|b|+b^{2}|<1$\\

iff $y<{1\over |(x_{0}-a)^{2}+(z_{0}-c)^{2}+2|b|+b^{2}|}$\\

and, applying Newton's theorem, with $b_{n}={(-1)^{n}(2n)!\over 2^{2n}(n!)^{2}}$;\\

$f_{x_{0},z_{0}}({1\over y})=y\sum_{n=0}^{\infty}b_{n}(y^{2}(x_{0}-a)^{2}+y^{2}(z_{0}-c)^{2}-2yb+y^{2}b^{2})^{n}$\\

$=\sum_{n=1}^{\infty}a_{n}y^{n}$\\

is an absolutely convergent power series in $y$ of order $1$. A similar result holds for $y<0$, so that $f_{x_{0},z_{0}}$ is analytic at infinity. Similarly, fixing $y_{0}\in\mathcal{R}$, we can see that $f_{x_{0},y_{0}}(z)$ is analytic at infinity, and $(i),(ii)$ in Definition \ref{normaldef} hold for $f_{x_{0}}(y,z)$. That $f_{x_{0}}$ is of very moderate decrease will follow from the proof below that $f$ is of very moderate decrease. As if there exists a constant $C\in\mathcal{R}_{>0}$ such that $|f|(x,y,z)\leq {C\over |x,y,z|}$, for $|(x,y,z)|\geq 1$, then if $|(y,z)|\geq 1$, we have that $|x_{0},y,z|\geq |(y,z)|\geq 1$,  and;\\

$|f_{x_{0}}|(y,z)\leq {C\over |(x_{0},y,z)|}\leq {C\over |(y,z)|}$\\

Similarly, the proof that ${\partial f_{x_{0}}\over \partial y}$ and ${\partial f_{x_{0}}\over \partial z}$ are of moderate decrease, will follow from the proof below that the higher derivatives ${\partial^{i+j+k}f\over \partial x^{i}\partial y^{j}\partial z^{k}}$ are of moderate decrease $i+j+k+1$, As if ${\partial f\over \partial y}$ is of moderate decrease $2$, so of moderate decrease, than, for $|(x,y,z)|>1$;\\

$|{\partial f\over \partial y}|\leq {C\over |(x,y,z)|^{2}}$\\

so that, for $|(y,z)|\geq 1$, as above;\\

$|{\partial f_{x_{0}}\over \partial y}|\leq {C\over |(x_{0},y,z)|^{2}}\leq {C\over |y,z|^{2}}$\\

The proof that for $y\in\mathcal{R}$, the zeros of $\{f_{x_{0},y}(z),f'_{x_{0},y}(z),f''_{x_{0},y}(z)\}$ are uniformly bounded in $y$ follows from the proof below that the zeros of $\{f_{x,y}(z),f'_{x,y}(z),f''_{x,y}(z),f'''_{x,y}(z),f''''_{x,y}(z)\}$ are uniformly bounded in $(x,y)$. It follows that for $x_{0}\in\mathcal{R}$, $f_{x_{0}}(y,z)$ is normal, similarly for $f_{y_{0}}(x,z)$ and $f_{z_{0}}(x,y)$, where $y_{0}\in\mathcal{R}$, $z_{0}\in\mathcal{R}$. We have then verified $(i)-(iii)$ of Definition \ref{normal3}. For $(iv)$, let $\overline{r}=(x,y,z)$, $r=|(x,y,z)|$, $\overline{a}=(a,b,c)$, then, switching to polars, $x=rsin(\theta)cos(\phi)$, $y=rsin(\theta)sin(\phi)$, $z=rcos(\theta)$, for $0\leq \theta\leq \pi$, $-\pi\leq \phi<\pi$;\\

$lim_{r\rightarrow\infty}r f(r,\theta,\phi)=lim_{r\rightarrow \infty}{r\over |\overline{r}(r,\theta,\phi)-\overline{a}|}$\\

$=lim_{r\rightarrow\infty}{1\over |{\overline{r}(r,\theta,\phi)\over r}-{\overline{a}\over r}|}$\\

$=lim_{r\rightarrow\infty}{1\over |\hat{\overline{r}}(r,\theta,\phi)|}$\\

$=1$\\

so that fixing a closed ball $B(\overline{0},s)\supset W$, using the fact that $f$ is smooth on $B(\overline{0},s)^{c}$, $r|f|\leq D$, where $D\in\mathcal{R}_{>0}$, on $B(\overline{0},s)^{c}$, so that $|f|\leq {D\over r}$, for $r>s$. As $f$ is continuous on ${B(\overline{0},s)\setminus W^{\circ}}$, it is bounded, by compactness of ${B(\overline{0},s)\setminus W^{\circ}}$, so that $|f|\leq M$ for $|\overline{x}|\geq 1$. It follows that $|f|\leq {C\over |\overline{x}|}$, for $|\overline{x}|\geq 1$, where $C=max(D,Ms)$. Therefore, $f|_{{\mathcal{R}^{3}\setminus W}}$ is of very moderate decrease.\\

Suppose inductively, that, for $(i+j+k\geq 0)$, ${\partial f^{i+j+k}\over \partial x^{i}\partial y^{j}\partial z^{k}}$ is of the form;\\

${p(x-a,y-b,z-c)\over |(x,y,z)-(a,b,c)|^{1+2(i+j+k)}}$, where $p$ is homogeneous of degree $i+j+k$, then, using the product rule;\\

${\partial f^{i+j+k+1}\over \partial x^{i}\partial y^{j}\partial z^{k+1}}={|(x,y,z)-(a,b,c)|^{2}{\partial p\over \partial z}(x-a,y-b,z-c)-(z-c)(1+2(i+j+k))p(x-a,y-b,z-c)\over |(x,y,z)-(a,b,c)|^{1+2(i+j+k+1)}}$\\

which is of the form;\\

${q(x-a,y-b,z-c)\over |(x,y,z)-(a,b,c)|^{1+2(i+j+k+1)}}$, where $q$ is homogeneous of degree $i+j+k+1$;\\

 as $(1+2(i+j+k))> i+j+k$, for $i+j+k\geq 0$\\

 Similar inductions work for ${\partial \over \partial x}$ and ${\partial \over \partial y}$, so that we can assume that;\\

${\partial f^{i+j+k}\over \partial x^{i}\partial y^{j}\partial z^{k}}$ is of the form;\\

${p(x-a,y-b,z-c)\over |(x,y,z)-(a,b,c)|^{1+2(i+j+k)}}$, where $p$ is homogeneous of degree $i+j+k$\\

We then have that;\\

$lim_{r\rightarrow\infty}r^{1+i+j+k}|{\partial f^{i+j+k}\over \partial x^{i}\partial y^{j}\partial z^{k}}|=lim_{r\rightarrow\infty}{r^{1+i+j+k}\sum_{i'+j'+k'=i+j+k}c_{i'j'k'}(x-a)^{i'}(y-b)^{j'}(z-c)^{k'}\over |\overline{r}-\overline{a}|^{1+2(i+j+k)}}$\\

$\leq lim_{r\rightarrow\infty}{r^{1+i+j+k}(\sum_{i'+j'+k'}|c_{i'j'k'}|)|\overline{r}-\overline{a}|^{i+j+k}\over |\overline{r}-\overline{a}|^{1+2(i+j+k)}}$\\

$=Elim_{r\rightarrow\infty}{r^{1+i+j+k}\over |\overline{r}-\overline{a}|^{1+i+j+k}}$\\

$=Elim_{r\rightarrow\infty}{1\over |{\overline{r}\over r}-{\overline{a}\over r}|^{1+i+j+k}}$\\

$=Elim_{r\rightarrow\infty}{1\over |\hat{\overline{r}}|^{1+i+j+k}}$\\

$=E$\\

where $E=\sum_{i'+j'+k'=i+j+k}|c_{i'j'k'}|$, so that fixing a closed ball $B(\overline{0},s)\supset W$, using the fact that ${\partial^{i+j+k} f\over \partial x^{i}\partial y^{j}\partial z^{k}}$ is smooth on $B(\overline{0},s)^{c}$, $r^{i+j+k+1}|{\partial f^{i+j+k}\over \partial x^{i}\partial y^{j}\partial z^{k}}|\leq D$, where $D\in\mathcal{R}_{>0}$, on $B(\overline{0},s)^{c}$, so that $|{\partial f^{i+j+k}\over \partial x^{i}\partial y^{j}\partial z^{k}}|\leq {D\over r^{i+j+k+1}}$, for $r>s$. As ${\partial^{i+j+k} f\over \partial x^{i}\partial y^{j}\partial z^{k}}$ is continuous on ${B(\overline{0},s)\setminus W^{\circ}}$, it is bounded, by compactness of ${B(\overline{0},s)\setminus W^{\circ}}$ again, so that $|{\partial^{i+j+k} f\over \partial x^{i}\partial y^{j}\partial z^{k}}|\leq M$ for $|x|\geq 1$. It follows that $|{\partial^{i+j+k} f\over \partial x^{i}\partial y^{j}\partial z^{k}}|\leq {C\over |\overline{x}|^{i+j+k+1}}$, for $|\overline{x}|\geq 1$, where $C=max(D,Ms^{i+j+k+1})$. It follows that $\{{\partial^{i+j+k} f\over \partial x^{i}\partial y^{j}\partial z^{k}}_{{\mathcal{R}^{3}\setminus W}}:i+j+k\geq 1\}$ are of moderate decrease $i+j+k+1$.\\

Fixing $\{y,z\}\subset\mathcal{R}$, we have that $f_{y,z}$ has no zeros, and, by the chain rule;\\

${\partial f\over \partial x}|_{y,z}=-{1\over 2}2(x-a){1\over ((x-a)^{2}+(y-b)^{2}+(z-c)^{2})^{3\over 2}}$\\

$=-{x-a\over ((x-a)^{2}+(y-b)^{2}+(z-b)^{2})^{3\over 2}}$\\

so that $(f_{y,z})'$ has a zero when $x=a$. We have that, by the product rule;\\

${\partial^{2} f\over \partial x^{2}}|_{y,z}=-{1\over ((x-a)^{2}+(y-b)^{2}+(z-c)^{2})^{3\over 2}}+{3(x-a)^{2}\over ((x-a)^{2}+(y-b)^{2}+(z-c)^{2})^{5\over 2}}=0$\\

iff $3(x-a)^{2}-[(x-a)^{2}+(y-b)^{2}+(z-c)^{2}]=0$\\

iff $2(x-a)^{2}=(y-b)^{2}+(z-c)^{2}$\\

iff $x-a={1\over \sqrt{2}}[(y-b)^{2}+(z-c)^{2}]^{1\over 2}$ or $x-a=-{1\over \sqrt{2}}[(y-b)^{2}+(z-c)^{2}]^{1\over 2}$\\

so that $(f_{y,z})''$ has at most $2$ zeros for $(y,z)\in\mathcal{R}^{2}$\\

Similarly;\\

${\partial^{3} f\over \partial x^{3}}|_{y,z}={9(x-a)\over ((x-a)^{2}+(y-b)^{2}+(z-c)^{2})^{5\over 2}}-{15(x-a)^{3}\over ((x-a)^{2}+(y-b)^{2}+(z-c)^{2})^{7\over 2}}=0$\\

iff $9(x-a)[(x-a)^{2}+(y-b)^{2}+(z-c)^{2}]-15(x-a)^{3}=0$\\

iff $9[(x-a)^{2}+(y-b)^{2}+(z-c)^{2}]-15(x-a)^{2}=0$\\

iff $6(x-a)^{2}=9((y-b)^{2}+(z-c)^{2})$\\

iff $x-a=\sqrt{3\over 2}[(y-b)^{2}+(z-c)^{2}]^{1\over 2}$ or $x-a=-\sqrt{3\over 2}[(y-b)^{2}+(z-c)^{2}]^{1\over 2}$\\

so that $(f_{y,z})'''$ has at most $2$ zeros for $(y,z)\in\mathcal{R}^{2}$.\\

Finally;\\

${\partial^{4} f\over \partial x^{4}}|_{y,z}={9\over ((x-a)^{2}+(y-b)^{2}+(z-c)^{2})^{5\over 2}}-{90(x-a)^{2}\over ((x-a)^{2}+(y-b)^{2}+(z-c)^{2})^{7\over 2}}+{105(x-a)^{4}\over ((x-a)^{2}+(y-b)^{2}+(z-c)^{2})^{9\over 2}}$\\

$=0$\\

iff $9[(x-a)^{2}+(y-b)^{2}+(z-c)^{2}]^{2}-90(x-a)^{2}[(x-a)^{2}+(y-b)^{2}$\\

$+(z-c)^{2}]+105(x-a)^{4}=0$\\

iff $9(u^{2}+s^{2})^{2}-90u^{2}(u^{2}+s^{2})+105u^{4}=0$\\

iff $24u^{4}-72u^{2}s^{2}+9s^{4}=0$\\

iff $8u^{4}-24u^{2}s^{2}+3s^{4}=0$\\

iff $u^{2}={24s^{2}+/-\sqrt{24^{2}s^{4}-4.24s^{4}}\over 16}$\\

iff $u^{2}={24s^{2}+/-\sqrt{480s^{4}}\over 16}$\\

iff $u^{2}={6+/-\sqrt{30}\over 4}s^{2}$\\

iff $u=\sqrt{6+\sqrt{30}}s$ or $u=\sqrt{6-\sqrt{30}}s$ or $u=-(\sqrt{6+\sqrt{30}})s$\\

 or $u=-(\sqrt{6-\sqrt{30}})s$\\

where $u=x-a$, $s=[(y-b)^{2}+(z-c)^{2}]^{1\over 2}$, so that $(f_{y,z})''''$ has at most $4$ zeros for $(y,z)\in\mathcal{R}^{2}$, and we can take $val=4$. A similar result holds for $f_{x,y}$ and $f_{x,z}$. It follows that $f$ is normal.\\

For the second claim, note that if $(x,y,z)\in Supp(\rho)$, then, switching to polars;\\

$f(x,y,z)=\int_{\mathcal{R}^{2}}{\rho(x',y',z')\over |(x,y,z)-(x',y',z')|}dx'dy'dz'$\\

$=\int_{\mathcal{R}^{3}}{\rho(x-x',y-y',z-z')\over |(x',y',z')|}dx'dy'dz'$\\

$=\int_{0,\pi}\int_{-\pi}^{\pi}\int_{\mathcal{R}_{>0}}{\rho_{x,y,z}(r,\theta,\phi)\over r}r^{2}sin(\theta)drd\theta d\phi$\\

$=\int_{0,\pi}\int_{-\pi}^{\pi}\int_{\mathcal{R}_{>0}}\rho_{x,y,z}(r,\theta,\phi)rsin(\theta)drd\theta d\phi$\\

so that;\\

$|f(x,y,z)|\leq \int_{0,\pi}\int_{-\pi}^{\pi}\int_{\mathcal{R}_{>0}}|\rho_{x,y,z}(r,\theta,\phi)|drd\theta d\phi$\\

$\leq 2M\pi^{2} R(x,y,z)$\\

where $Supp_{x,y,z}(\rho)\subset B(\overline{0},R(x,y,z))$, $\rho|\leq M$, so that $f$ is defined everywhere. If $\rho$ is smooth, we have that $f$ is smooth, as;\\

${\partial^{i+j+k}f\over \partial x^{i}\partial y^{j}\partial z^{k}}(x,y,z)={\partial^{i+j+k} \int_{\mathcal{R}^{3}}{\rho(x-x',y-y',z-z')\over |(x',y',z')|}dx'dy'dz'\over \partial x^{i}\partial y^{j}\partial z^{k}}$\\

$=\int_{\mathcal{R}^{3}}{{\partial^{i+j+k}\rho\over \partial x^{i}\partial y^{j}\partial z^{k}}(x-x',y-y',z-z')\over |(x',y',z')|}dx'dy'dz'$\\

with ${\partial^{i+j+k}\rho\over \partial x^{i}\partial y^{j}\partial z^{k}}$ having compact support again. If $\rho$ is continuous, but not necessarily smooth, we have that, for $(x,y,z)\notin Supp(\rho)$;\\

${\partial^{i+j+k}f\over \partial x^{i}\partial y^{j}\partial z^{k}}(x,y,z)=\int_{\mathcal{R}^{3}}\rho(x',y',z'){\partial^{i+j+k}\over \partial x^{i}\partial y^{j}\partial z^{k}}({1\over |(x,y,z)-(x',y',z')|})dx'dy'dz'$\\

so that $f$ is smooth, outside $Supp(\rho)$. When $\rho\geq 0$ with compact support, we have to show that $f_{x_{0}}(y,z)$ is quasi normal. For $(i)$, in Definition \ref{normaldef}, we have, for $x_{0}\in\mathcal{R},z_{0}\in\mathcal{R}$, $y>0$;\\

$f_{x_{0},z_{0}}({1\over y})=\int_{\mathcal{R}^{3}}{\rho(x',y',z')\over |(x_{0},{1\over y},z)-(x',y',z')|}dx'dy'dz'$\\

$=y\int_{\mathcal{R}^{3}}{\rho(x',y',z')\over (1+y^{2}(x_{0}-x')^{2}+y^{2}(z_{0}-z')^{2}-2yy'+y^{2}y'^{2})^{1\over 2}}$\\

so that, with $y<1$, $y^{2}<y$, letting;\\

$M_{x_{0},z_{0}}=sup_{(x',y',z')\in Supp(\rho)}|(x_{0}-x')^{2}+(z_{0}-z')^{2}+2|y'|+y'^{2}|$, if $y<{1\over M_{x_{0},z_{0}}}$, then;\\

$|y^{2}(x_{0}-x')^{2}+y^{2}(z_{0}-z')^{2}-2yy'+y^{2}y'^{2}|<y|(x_{0}-x')^{2}+(z_{0}-z)^{2}+2|y'|+y'^{2}|<1$\\

 so that, we can apply Newton's theorem uniformly in $(x',y',z')\in Supp(\rho)$, to obtain;\\

$f_{x_{0},z_{0}}({1\over y})=y\int_{Supp(\rho)}\rho(x',y',z')(\sum_{n=0}^{\infty}b_{n}(y^{2}(x_{0}-x')^{2}+y^{2}(z_{0}-z')^{2}-2yy'+y^{2}y'^{2})^{n})dx'dy'dz'$\\

where $b_{n}$ is as above. With $y<1$ again, $|x'|\leq M$, $|y'|\leq M$, $|z'|\leq M$ for $(x',y',z')\subset Supp(\rho)$, $|\rho|\leq N$, $y<{1\over (|x_{0}|+M)^{2}+(|z_{0}+M)^{2}+2M+M^{2}}$, we have, applying the DCT;\\

$|f_{x_{0},z_{0}}({1\over y})|\leq yN(2M)^{3}\sum_{n=0}^{\infty}|b_{n}|((|x_{0}|+M)^{2}y+(|z_{0}|+M)^{2}+2yM+yM^{2})^{n}$\\

$\leq 8yNM^{3}\sum_{n=0}^{\infty}|b_{n}|y^{n}((|x_{0}|+M)^{2}+(|z_{0}+M)^{2}+2M+M^{2})^{n}$\\

$\leq 8yNM^{3}\sum_{n=0}^{\infty}y^{n}((|x_{0}|+M)^{2}+(|z_{0}|+M)^{2}+2M+M^{2})^{n}$\\

defines an absolutely convergent series. A similar proof works for $y<0$. $(ii)$ is similar. For $(iii)$, in Definition \ref{normaldef}, this will follow, as above, from the main result that $f$ itself is of very moderate decrease, similarly, for $(iv)$, as above, the moderate decrease in the fibre ${\partial f\over \partial y}_{x_{0}}$ and ${\partial f\over \partial z}_{x_{0}}$, will follow from the moderate decrease $i+j+k+1$ in the derivatives ${\partial^{i+j+k}f\over \partial x^{i}\partial y^{j}\partial z^{k}}$, $i+j+k\geq 1$. For $(v)$, again, as above, the claim on the intervals and zeros follows from the main proof.\\

For $(iv)'$ in Definition \ref{normal3}, with $Supp(\rho)\subset B(\overline{0},M)$, $M>1$, $|\rho|\leq N$, if $|(x,y,z)|\geq 2M$, and $(x',y',z')\in Supp(\rho)$, $|(x,y,z)-(x',y',z)|\geq {|(x,y,z)\over 2}$, so that ${1\over |(x,y,z)-(x',y',z')|}\leq {2\over |(x,y,z)|}$, and;\\

$|f(x,y,z)|= |\int_{Supp(\rho)}{\rho(x',y',z')\over |(x,y,z)-(x',y',z')|}dx'dy'dz'|$\\

$\leq {2\over |(x,y,z)|}\int_{Supp(\rho)}|\rho(x',y',z')|dx'dy'dz'$\\

$\leq {8\pi M^{3}N\over 3|(x,y)|}$\\

For $(v)'$, we have that, combining results above, with $i+j+k\geq 1$, $(x,y,z)\notin Supp(\rho)$;\\

${\partial^{i+j+k} f\over \partial x^{i}\partial y^{j}\partial z^{k}}(x,y,z)=\int_{\mathcal{R}^{3}}{\rho(x',y',z')p(x-x',y-y',z-z')\over |(x,y,z)-(x',y',z')|^{1+2(i+j+k)}}dx'dy'dz'$\\

where $p$ is homogeneous of degree $i+j+k$. Then, if $|(x,y,z)|\geq max(2M,1)$, and $(x',y',z')\in Supp(\rho)$;\\

${1\over |(x,y,z)-(x',y',z')|^{1+2(i+j+k)}}\leq {2^{1+2(i+j+k)}\over |(x,y,z)|^{1+2(i+j+k)}}$\\

$|x-x'|\leq |(x,y,z)-(x',y',z')|\leq |(x,y,z)|+M$, and similarly $|y-y'|\leq |(x,y,z)|+M$, $|z-z'|\leq |(x,y,z)|+M$, so that;\\

$|p(x-x',y-y',z-z')|\leq T(|(x,y,z)|+M)^{i+j+k}$\\

where $T=\sum_{i'+j'+k'=i+j+k}|a_{i'j'k'}|$ and $p=\sum_{i'+j'+k'=i+j+k}a_{i'j'k'}x_{1}^{i'}x_{2}^{j'}x_{3}^{k'}$. It follows that;\\

$|{{\partial^{i+j+k} f\over \partial x^{i}\partial y^{j}\partial z^{k}}}|=|\int_{Supp(\rho)}{\rho(x',y',z')p(x-x',y-y',z-z')\over |(x,y,z)-(x',y',z')|^{1+2(i+j+k)}}dx'dy'dz'|$\\

$\leq {2^{1+2(i+j+k)}T(|(x,y,z)|+M)^{i+j+k}\over |(x,y,z)|^{1+2(i+j+k)}}\int_{Supp(\rho)}|\rho(x',y',z')|dx'dy'dz'$\\

$\leq {2^{1+2(i+j+k)}TN{4\pi M^{3}\over 3}(|(x,y,z)|+M)^{i+j+k}\over |(x,y,z)|^{1+2(i+j+k)}}$\\

$\leq {2^{1+2(i+j+k)}TN{4\pi M^{3}\over 3}(i+j+k+1)!M^{i+j+k} \over |(x,y,z)|^{1+(i+j+k)}}$\\

so that ${\partial^{i+j+k} f\over \partial x^{i}\partial y^{j}\partial z^{k}}(x,y,z)$ is of moderate decrease $i+j+k+1$.\\

For $(vi)'$, we have that if $\rho\geq 0$ is continuous with compact support, $\rho\neq 0$, that $f>0$. Repeating the argument above, and the calculation of the derivatives, fixing $(y,z)\in\mathcal{R}^{2}$, and letting $s=(x^{2}+y^{2})^{1\over 2}$, we see that the zeros of $f_{x,y}'$ are contained in the interval $(-(M+1), M+1)$, with the length $2(M+1)$ of the interval, uniformly bounded in $(y,z)$, the zeros of $(f_{x,y})''$ are contained in the intervals $(-M+{1\over \sqrt{2}}(s-M),M+{1\over \sqrt{2}}(s+M))\cup (-M-{1\over \sqrt{2}}(s+M),M-{1\over \sqrt{2}}(s-M))$, with the length of the intervals, $(2+\sqrt{2})M$ uniform in $s>2M$, the zeros of $(f_{x,y})''$ are contained in the intervals $(-M+{\sqrt{3}\over \sqrt{2}}(s-M),M+{\sqrt{3}\over \sqrt{2}}(s+M))\cup (-M-{\sqrt{3}\over \sqrt{2}}(s+M),M-{\sqrt{3}\over \sqrt{2}}(s-M))$, with the length of the intervals, $(2+\sqrt{6})M$ uniform in $s>2M$, the zeros of $(f_{x,y})'''$ are contained in the intervals $(-M+\sqrt{6+\sqrt{30}}(s-M),M+\sqrt{6+\sqrt{30}}(s+M))\cup (-M-\sqrt{6+\sqrt{30}}(s+M),M-\sqrt{6+\sqrt{30}}(s-M))$, $(-M+\sqrt{6-\sqrt{30}}(s-M),M+\sqrt{6-\sqrt{30}}(s+M))\cup (-M-\sqrt{6-\sqrt{30}}(s+M),M-\sqrt{6-\sqrt{30}}(s-M))$ with the length of the intervals, $(2+2\sqrt{6+\sqrt{30}})M$ and $(2+2\sqrt{6-\sqrt{30}})M$ uniform in $s>2M$.\\

Again the proof that if $\rho\neq 0$ is smooth with compact support, then $f$ is quasi split normal, follows easily by observing that $\rho=\rho^{+}+\rho^{-}$, with $\{\rho^{+},\rho^{-}\}$ being continuous, $\rho^{+}\geq 0$, $\rho^{-}\leq 0$, and using the proof for quasi normality, along with the previous observation in dimension $2$, that quasi normality implies quasi split normality.\\

\end{proof}
\begin{lemma}
\label{normal4}

Let $f:\mathcal{R}^{3}\rightarrow\mathcal{R}$ be normal, then, for $\{x,y,z\}\subset\mathcal{R}$, $k_{1}\neq 0$, $k_{2}\neq 0$, $k_{3}\neq 0$;\\

$A(k_{1},y,z)=lim_{r\rightarrow\infty}\int_{-r}^{r}f(x,y,z)e^{-ik_{1}x}dx$\\

$B(x,k_{2},z)=lim_{r\rightarrow\infty}\int_{-r}^{r}f(x,y,z)e^{-ik_{2}y}dy$\\

$C(x,y,k_{3})=lim_{r\rightarrow\infty}\int_{-r}^{r}f(x,y,z)e^{-ik_{3}z}dz$\\

all exist and $A(k_{1},y,z),B(x,k_{2},z),C(x,y,k_{3})$ are of moderate decrease $3$.\\

and, for $\{x,y,z\}\subset\mathcal{R}$, $k_{1}\neq 0$, $k_{2}\neq 0$, $k_{3}\neq 0$;\\

$F(k_{1},k_{2},z)=lim_{r,s\rightarrow\infty}\int_{-r}^{r}\int_{-s}^{s}f(x,y,z)e^{-ik_{1}x}e^{-ik_{2}y}dxdy$\\

$G(k_{1},y,k_{3})=lim_{r,s\rightarrow\infty}\int_{-r}^{r}\int_{-s}^{s}f(x,y,z)e^{-ik_{1}x}e^{-ik_{3}z}dxdz$\\

$H(x,k_{2},k_{3})=lim_{r,s\rightarrow\infty}\int_{-r}^{r}\int_{-s}^{s}f(x,y,z)e^{-ik_{2}y}e^{-ik_{3}z}dydz$\\

all exist and $F(k_{1},k_{2},z),G(k_{1},y,k_{3}),H(x,k_{2},k_{3})$ are of moderate decrease.\\

Moreover;\\

$F(k_{1},k_{2},z)=\int_{-\infty}^{\infty}A(k_{1},y,z)e^{-ik_{2}y}dy$\\

and corresponding results hold for $\{A,B,C,F,G,H\}$, integrating out the variables in a similar way.\\

\end{lemma}
\begin{proof}
The first claim follows from \cite{dep} together with Lemma \ref{modulus3} and $(i),(ii),(iii)$ in Definition \ref{normal3}, using the fact that, by normality, $f_{x,y}(z),f_{x,z}(y),f_{y,z}(x)$ are analytic at infinity, for $\{(x,y),(x,z),(y,z)\}\subset\mathcal{R}^{2}$, $(*)$. We then have, using integration by parts, for $(y_{0},z_{0})\in\mathcal{R}^{2}$;\\

$\int_{-r}^{r}f(x,y_{0},z_{0})e^{-ik_{1}x}dx=[{if(x,y_{0})e^{-ik_{1}x}\over k_{1}}]_{-r}^{r}-\int_{-r}^{r}{i{\partial f\over \partial x}(x,y_{0},z_{0})e^{-ik_{1}x}\over k_{1}}dx$\\

so that, as $f_{y_{0},z_{0}}$ is of very moderate decrease, by $(iv)$ in Definition \ref{normal3}, integrating by parts;\\

$A(k_{1},y_{0},z_{0})=lim_{r\rightarrow\infty}\int_{-r}^{r}f(x,y_{0},z_{0})e^{-ik_{1}x}dx$\\

$=lim_{r\rightarrow\infty}([{if(x,y_{0},z_{0})e^{-ik_{1}x}\over k_{1}}]_{-r}^{r}-\int_{-r}^{r}{i{\partial f\over \partial x}(x,y_{0},z_{0})e^{-ik_{1}x}\over k_{1}}dx)$\\

$=lim_{r\rightarrow\infty}(-\int_{-r}^{r}{i{\partial f\over \partial x}(x,y_{0},z_{0})e^{-ik_{1}x}\over k_{1}}dx)$\\

$=-{i\over k_{1}}\int_{-\infty}^{\infty}{\partial f\over \partial x}(x,y_{0},z_{0})e^{-ik_{1}x}dx$\\

the last integral being definite, as $({\partial f\over \partial x})_{y_{0},z_{0}}$ is of moderate decrease, using $(v)$ in Definition \ref{normal3}. It follows that $A(k_{1},y,z)$ is smooth, as differentiating under the integral sign is justified by the DCT, MVT and $(v)$ again, with;\\

${\partial^{j+k} A\over \partial y^{j}\partial z^{k}}(k_{1},y,z)=-{i\over k_{1}}\int_{-\infty}^{\infty}{\partial^{1+j+k}f\over \partial x\partial y^{j}\partial z^{k}}(x,y,z)e^{-ik_{1}x}dx$\\

Integrating by parts again;\\

$A(k_{1},y_{0},z_{0})=-{1\over k_{1}^{2}}\int_{-\infty}^{\infty}{\partial^{2} f\over \partial x^{2}}(x,y_{0},z_{0})e^{-ik_{1}x}dx$\\

with ${\partial^{2} f\over \partial x^{2}}$ of moderate decrease $3$, by $(v)$ again . We have that $({\partial^{2} f\over \partial x^{2}})_{y_{0},z_{0}}$ is analytic at infinity by $(*)$ and Lemma \ref{modulus2}. By $(vi)$ in the Definition \ref{normal3}, we can assume that for sufficiently large $(y,z)$, $({\partial^{2} f\over \partial x^{2}})_{y,z}$ is monotone and positive/negative in the intervals $(-\infty,a_{1}(y,z)),\ldots,(a_{val}(y,z),\infty)$, where $a_{1},\ldots,a_{val}$ vary continuously with $y,z$. Splitting the integral into $cos(k_{1}x)$ and $sin(k_{1}x)$ components, in a similar calculation to \cite{dep2}, for any interval of length at least ${2\pi\over |k_{1}|}$, we obtain an alternating cancellation in the contribution to the integral of at most ${2\pi ||f_{y}'||\over |k_{1}|}$ and for any interval of length at most ${2\pi\over |k_{1}|}$, we obtain a contribution to the integral of at most ${4\pi ||f_{y}'||\over |k_{1}|}$. It follows that, for sufficiently large $(y,z)$, using the fact that ${\partial^{2} f\over \partial x^{2}}$ has moderate decrease $3$;\\

$|-{i\over k_{1}}\int_{-\infty}^{\infty}{\partial^{2} f\over \partial x^{2}}(x,y,z)e^{-ik_{1}x}dx|$\\

$\leq {(val+1)\over |k_{1}|}{4\pi ||f_{y,z}''||\over |k_{1}|}$\\

$={4\pi(val+1)||f_{y,z}''||\over |k_{1}|^{2}}$\\

$\leq {4\pi(val+1)C\over |k_{1}|^{2}|(y,z)|^{3}}$ $(*)$\\

so that;\\

$|A(k_{1},y,z)|\leq {D\over |y,z|^{3}}$ $(**)$\\

for sufficiently large $(y,z)$, with $D={4\pi(val+1)C\over |k_{1}|^{2}}$. As $A(k_{1},y,z)$ is smooth, we obtain that $A(k_{1},y,z)$ is of moderate decrease $3$. Similarly, we can show that $B(x,k_{2},z)$ and $C(x,y,k_{3})$ are of moderate decrease $3$, for $k_{2}\neq 0$, $k_{3}\neq 0$.\\

The second claim follows from Lemma \ref{normal1}, using normality of the fibres, $(i),(ii),(iii)$ in Definition \ref{normal3}. In fact, the first integral is indefinite, in the sense that we could define it as;\\

$F(k_{1},k_{2},z)=lim_{r_{1}\rightarrow\infty, r_{2}\rightarrow\infty, s_{1}\rightarrow\infty,s_{2}\rightarrow\infty}(\int_{-r_{1}}^{a}+\int_{a}^{r_{2}})(\int_{-s_{1}}^{b}+\int_{a}^{s_{2}})f(x,y,z)e^{-ik_{1}x}e^{-ik_{2}y}dxdy$\\

for a choice of $a,b\subset\mathcal{R}$, and similarly for $G(k_{1},y,k_{3}),H(x,k_{2},k_{3})$. Now observe that;\\

$F(k_{1},k_{2},z)=\int_{-\infty}^{\infty}A(k_{1},y,z)e^{-ik_{2}y}dy$\\

where $A_{k_{1}}$ is from the first part of the lemma, and of moderate decrease $3$. This follows from the result of Lemma \ref{normal1}, and the fact that the fibre $f_{z}$ is normal. We claim that $F(k_{1},k_{2},z)$ is of moderate decrease $2$. We have that $|A_{k_{1}}(y,z)|\leq {C_{k_{1}}\over |y,z|^{3}}$ for $|(y,z)>1$, by $(**)$. It follows that, for sufficiently large $z$;\\

$|F(k_{1},k_{2},z)|\leq |\int_{-\infty}^{\infty}A(k_{1},y,z)e^{-ik_{2}y}dy|$\\

$\leq \int_{-\infty}^{\infty}|{C_{k_{1}}\over (y^{2}+z^{2})^{3\over 2}}|dy$\\

$={C_{k_{1}}\over |z|^{3}}\int_{-\infty}^{\infty}|{C_{k_{1}}\over (1+{y^{2}\over z^{2}})^{3\over 2}}|dy$\\

$={C_{k_{1}}\over |z|^{3}}\int_{-{\pi\over 2}}^{{\pi\over 2}}{1\over sec^{3}(\theta)}zsec^{2}(\theta)d\theta$, ($(tan(\theta)={y\over z}), dy=zsec^{2}(\theta)d\theta$)\\

$={C_{k_{1}}\over |z|^{2}}\int_{-{\pi\over 2}}^{{\pi\over 2}}{1\over sec(\theta)}d\theta$\\

$\leq {C_{k_{1}}\over |z|^{2}}\int_{-{\pi\over 2}}^{{\pi\over 2}}cos(\theta)d\theta$\\

$\leq {\pi C_{k_{1}}\over |z|^{2}}$\\

so that $F(k_{1},k_{2},z)$ is of moderate decrease. Similar results hold for $G(x,k_{2},z)$ and $H(x,y,k_{3})$, with $k_{1}\neq 0$, $k_{2}\neq 0$, $k_{3}\neq 0$.

\end{proof}
\begin{lemma}
\label{normal4quasi}
Let $f:\mathcal{R}^{3}\rightarrow\mathcal{R}$, the same result as Lemma \ref{normal4} holds, if $f$ is quasi normal or $f$ is quasi split normal.\\

\end{lemma}
\begin{proof}
Again, we just have to replace the uses of $(i)-(vi)$ in Definition \ref{normal3}, within Lemma \ref{normal4}, with the use of $(i)-(vi)'$ or $(i)-(vi)''$. The method of replacing $(vi)$ by $(vi)'$ is given in Lemma \ref{quasinormal}. The fact that we can replace $(i),(ii),(iii)$ by $(i)',(ii),(iii)'$ at the beginning of the proof of the second claim follows from Lemma \ref{quasinormal1}, and at the beginning of the proof from Definition \ref{normaldef}. The use of $(iv),(v)$ and $(iv)'(v')$ is the same. A similar argument works in the quasi split normal case, using the argument at the end of Lemma \ref{quasinormal}.

\end{proof}
\begin{lemma}
\label{normal5}
Let hypotheses and notation be as in the previous lemma, then we can define, for $k_{1}\neq 0$, $k_{2}\neq 0$, $k_{3}\neq 0$;\\

$A(k_{1},k_{2},k_{3})=\int_{-\infty}^{\infty}\int_{-\infty}^{\infty}A(k_{1},y,z)e^{-ik_{2}y}e^{-ik_{3}z}dydz$\\

$B(k_{1},k_{2},k_{3})=\int_{-\infty}^{\infty}\int_{-\infty}^{\infty}B(x,k_{2},z)e^{-ik_{1}x}e^{-ik_{3}z}dxdz$\\

$C(k_{1},k_{2},k_{3})=\int_{-\infty}^{\infty}\int_{-\infty}^{\infty}B(x,y,k_{3})e^{-ik_{1}x}e^{-ik_{2}y}dxdy$\\

$F(k_{1},k_{2},k_{3})=\int_{-\infty}^{\infty}F(k_{1},k_{2},z)e^{-ik_{3}z}dz$\\

$G(k_{1},k_{2},k_{3})=\int_{-\infty}^{\infty}G(k_{1},y,k_{3})e^{-ik_{2}y}dy$\\

$H(k_{1},k_{2},k_{3})=\int_{-\infty}^{\infty}H(x,k_{2},k_{3})e^{-ik_{1}x}dx$\\

We have that;\\

$A(k_{1},k_{2},k_{3})=B(k_{1},k_{2},k_{3})=C(k_{1},k_{2},k_{3})=F(k_{1},k_{2},k_{3})=G(k_{1},k_{2},k_{3})=H(k_{1},k_{2},k_{3})$ $\dag)$\\

$=lim_{r\rightarrow\infty,s\rightarrow\infty,t\rightarrow\infty}\int_{-r}^{r}\int_{-s}^{s}\int_{-t}^{t}f(x,y,z)e^{-ik_{1}x}e^{-ik_{2}y}e^{-ik_{3}z}dxdydz$ $(\dag\dag)$\\
\end{lemma}

\begin{proof}
The definitions follows from Lemma \ref{normal4}, using the fact that $A(k_{1},y,z)$ is of moderate decrease $3$ and smooth, so it belongs to $L^{1}(\mathcal{R}^{2})$. We can then use the usual Fourier transform. Similarly, for $B(x,k_{2},z)$ and $C(x,y,k_{3})$.\\

Similarly, as $F(k_{1},k_{2},z)$ is smooth and of moderate decrease, we can then define the usual Fourier transform, for $k_{3}\neq 0$;\\

$F(k_{1},k_{2},k_{3})=\int_{-\infty}^{\infty}F(k_{1},k_{2},z)e^{-ik_{3}z}dz$\\

It is clear that $(\dag)$ holds, from the last claim in Lemma \ref{normal4}, once we have shown that $F=G=H$ and $(\dag\dag)$.\\

We have that;\\

$|\int_{-r}^{r}(F(k_{1},k_{2},z)-\int_{-s}^{s}\int_{-t}^{t}f(x,y,z)e^{-ik_{1}x}e^{-ik_{2}y}dxdy)e^{-ik_{3}z}dz|$\\

$=|\int_{-r}^{r}(\int_{-\infty}^{\infty}\int_{-\infty}^{\infty}f(x,y,z)e^{-ik_{1}x}e^{-ik_{2}y}dxdy-\int_{-s}^{s}\int_{-t}^{t}f(x,y,z)e^{-ik_{1}x}e^{-ik_{2}y}dxdy)e^{-ik_{3}z}dz|$\\

$=|\int_{-r}^{r}(\int_{(|x|\leq s,|y\leq t)^{c}}f(x,y,z)e^{-ik_{1}x}e^{-ik_{2}y}dxdy)e^{-ik_{3}z}dz|$\\

$=|\int_{-r}^{r}[\int_{|y|>t}\int_{-\infty}^{\infty}f(x,y,z)e^{-ik_{1}x}e^{-ik_{2}y}dxdy$\\

$+\int_{|x|>s}\int_{-\infty}^{\infty}f(x,y,z)e^{-ik_{1}x}e^{-ik_{2}y}dxdy$\\

$-\int_{|x|>s}\int_{|y|>t}f(x,y,z)e^{-ik_{1}x}e^{-ik_{2}y}dxdy]e^{-ik_{3}z}dz|$\\

$\leq |\int_{-r}^{r}(\int_{|y|>t}\int_{-\infty}^{\infty}f(x,y,z)e^{-ik_{1}x}e^{-ik_{2}y}dxdy)e^{-ik_{3}z}dz|$ $(i)$\\

$+|\int_{-r}^{r}(\int_{|x|>s}\int_{-\infty}^{\infty}f(x,y,z)e^{-ik_{1}x}e^{-ik_{2}y}dxdy)e^{-ik_{3}z}dz|$ $(ii)$\\

$+|\int_{-r}^{r}(\int_{|x|>s}\int_{|y|>t}f(x,y,z)e^{-ik_{1}x}e^{-ik_{2}y}dxdy)e^{-ik_{3}z}dz|$ $(iii)$\\

We estimate the three terms separately, using integration by parts, for $k_{1}\neq 0$, $k_{2}\neq 0$, $k_{3}\neq 0$. For $(i)$, as $f$ is of very moderate decrease, ${\partial^{3} f\over \partial x^{3}}$ is of moderate decrease $4$;\\

$|\int_{-r}^{r}(\int_{|y|>t}\int_{-\infty}^{\infty}f(x,y,z)e^{-ik_{1}x}e^{-ik_{2}y}dxdy)e^{-ik_{3}z}dz|$\\

$=|{-i\over k_{1}}\int_{-r}^{r}(\int_{|y|>t}\int_{-\infty}^{\infty}{\partial f\over \partial x}(x,y,z)e^{-ik_{1}x}e^{-ik_{2}y}dxdy)e^{-ik_{3}z}dz|$\\

$=|-{1\over k_{1}^{2}}\int_{-r}^{r}(\int_{|y|>t}\int_{-\infty}^{\infty}{\partial^{2} f\over \partial x^{2}}(x,y,z)e^{-ik_{1}x}e^{-ik_{2}y}dxdy)e^{-ik_{3}z}dz|$\\

$=|{i\over k_{1}^{3}}\int_{-r}^{r}(\int_{|y|>t}\int_{-\infty}^{\infty}{\partial^{3} f\over \partial x^{3}}(x,y,z)e^{-ik_{1}x}e^{-ik_{2}y}dxdy)e^{-ik_{3}z}dz|$\\

$\leq {1\over |k_{1}|^{3}}\int_{|y|>t}\int_{-\infty}^{\infty}\int_{\infty}^{\infty}|{\partial^{3} f\over \partial x^{3}}|dxdzdy$\\

$\leq {1\over |k_{1}|^{3}}\int_{|y|>t}\int_{-\infty}^{\infty}\int_{\infty}^{\infty}{C\over (x^{2}+y^{2}+z^{2})^{2}}dxdzdy$\\

$={1\over |k_{1}|^{3}}\int_{|y|>t}\int_{-\infty}^{\infty}{1\over (y^{2}+z^{2})^{2}}\int_{\infty}^{\infty}{C\over (1+{x^{2}\over y^{2}+z^{2}})^{2}}dxdzdy$\\

$={1\over |k_{1}|^{3}}\int_{|y|>t}\int_{-\infty}^{\infty}{1\over (y^{2}+z^{2})^{2}}\int_{-\pi\over 2}^{\pi\over 2}{C\over sec^{4}(\theta)}sec^{2}(\theta)(y^{2}+z^{2})^{1\over 2}d\theta dz dy$\\

$\leq {C\over k_{1}^{3}}\int_{|y|>t}\int_{-\infty}^{\infty}{1\over (y^{2}+z^{2})^{3\over 2}}\int_{-\pi\over 2}^{\pi\over 2}Ccos^{2}(\theta)d\theta dz dy$\\

$\leq {C\pi\over |k_{1}|^{3}}\int_{|y|>t}\int_{-\infty}^{\infty}{1\over (y^{2}+z^{2})^{3\over 2}}dzdy$\\

$={C\pi\over |k_{1}|^{3}}\int_{|y|>t}{1\over |y|^{3}} \int_{-\infty}^{\infty}{1\over (1+{z^{2}\over y^{2}})^{3\over 2}}dzdy$\\

$\leq {C\pi\over |k_{1}|^{3}}\int_{|y|>t}{1\over |y|^{2}} \int_{-{\pi\over 2}}^{{\pi\over 2}}cos^{2}(\theta)d\theta dy$\\

$\leq {C\pi^{2}\over |k_{1}|^{3}}\int_{|y|>t}{1\over |y|^{2}}dy$\\

$\leq {2C\pi^{2}\over |k_{1}|^{3}t}$ $(\dag)$\\

For $(ii)$, as $f$ is of very moderate decrease, ${\partial^{3} f\over \partial y^{3}}$ is of moderate decrease $4$, using the same argument;\\

$|\int_{-r}^{r}(\int_{|x|>s}\int_{-\infty}^{\infty}f(x,y,z)e^{-ik_{1}x}e^{-ik_{2}y}dxdy)e^{-ik_{3}z}dz|\leq {2C\pi^{2}\over |k_{2}|^{3}s}$\\

For $(iii)$, we have that, using integration by parts again;\\

$|\int_{-r}^{r}(\int_{|x|>s}\int_{|y|>t}f(x,y,z)e^{-ik_{1}x}e^{-ik_{2}y}dxdy)e^{-ik_{3}z}dz|$\\

$=|\int_{-r}^{r}(\int_{|x|>s}{i\over k_{2}}(-f(x,t,z)+f(x,-t,z))e^{-ik_{1}x}dx$\\

$-{i\over k_{2}}\int_{|x|>s}\int_{|y|>t}{\partial f\over \partial y}(x,y,z)e^{-ik_{1}x}e^{-ik_{2}y}dxdy)e^{-ik_{3}z}dz|$\\

$\leq |\int_{-r}^{r}\int_{|x|>s}{i\over k_{2}}(-f(x,t,z)+f(x,-t,z))e^{-ik_{1}x}e^{-ik_{3}z}dxdz|$\\

$+|\int_{-r}^{r}{i\over k_{2}}\int_{|x|>s}\int_{|y|>t}{\partial f\over \partial y}(x,y,z)e^{-ik_{1}x}e^{-ik_{2}y}dxdy)e^{-ik_{3}z}dz|$\\

$=|\int_{-r}^{r}{-1\over k_{1}k_{2}}(-f(s,-t,z)+f(-s,-t,z)+f(s,t,z)-f(-s,-t,z))e^{-ik_{3}z}dz$\\

$+\int_{-r}^{r}\int_{|x|>s}{1\over k_{1}k_{2}}(-{\partial f\over \partial x}(x,t,z)+{\partial f\over \partial x}(x,-t,z))e^{-ik_{1}x}e^{-ik_{3}z}dxdz|$\\

$+|\int_{-r}^{r}{i\over k_{2}}\int_{|x|>s}\int_{|y|>t}{\partial f\over \partial y}f(x,y,z)e^{-ik_{1}x}e^{-ik_{2}y}dxdy)e^{-ik_{3}z}dz|$\\

$\leq |\int_{-r}^{r}{-1\over k_{1}k_{2}}(-f(s,-t,z)+f(-s,-t,z)+f(s,t,z)-f(-s,-t,z))e^{-ik_{3}z}dz|$\\

$+|\int_{-r}^{r}\int_{|x|>s}{1\over k_{1}k_{2}}(-{\partial f\over \partial x}(x,t,z)+{\partial f\over \partial x}(x,-t,z))e^{-ik_{1}x}e^{-ik_{3}z}dxdz|$\\

$+|\int_{-r}^{r}{i\over k_{2}}\int_{|x|>s}\int_{|y|>t}{\partial f\over \partial y}f(x,y,z)e^{-ik_{1}x}e^{-ik_{2}y}dxdy)e^{-ik_{3}z}dz|$\\

$=|\int_{-r}^{r}{-1\over k_{1}k_{2}}(-f(s,-t,z)+f(-s,-t,z)+f(s,t,z)-f(-s,-t,z))e^{-ik_{3}z}dz|$\\

$+|\int_{-r}^{r}{i\over k_{1}^{2}k_{2}}({\partial f\over \partial x}(s,t,z)-{\partial f\over \partial x}(-s,t,z)-{\partial f\over \partial x}(s,-t,z)+{\partial f\over \partial x}(-s,-t,z))e^{-ik_{3}z}dz$\\

$-\int_{-r}^{r}{i\over k_{1}^{2}k_{2}}\int_{|x|>s}(-{\partial^{2}f\over \partial x^{2}}(x,t,z)+{\partial^{2}f\over \partial x^{2}}(x,-t,z))e^{-ik_{1}x}e^{-ik_{3}z}dxdz|$\\

$+|\int_{-r}^{r}{i\over k_{2}}\int_{|x|>s}\int_{|y|>t}{\partial f\over \partial y}(x,y,z)e^{-ik_{1}x}e^{-ik_{2}y}dxdy)e^{-ik_{3}z}dz|$\\

$\leq {1\over |k_{1}k_{2}|}|\int_{-r}^{r}(-f(s,-t,z)+f(-s,-t,z)+f(s,t,z)-f(-s,-t,z))e^{-ik_{3}z}dz|$\\

 $(a)$\\

$+{1\over |k_{1}^{2}k_{2}|}|\int_{-r}^{r}({\partial f\over \partial x}(s,t,z)-{\partial f\over \partial x}(-s,t,z)-{\partial f\over \partial x}(s,-t,z)+{\partial f\over \partial x}(-s,-t,z))e^{-ik_{3}z}dz|$\\

 $(b)$\\

$+{1\over |k_{1}^{2}k_{2}|}|\int_{-r}^{r}\int_{|x|>s}(-{\partial^{2}f\over \partial x^{2}}(x,t,z)+{\partial^{2}f\over \partial x^{2}}(x,-t,z))e^{-ik_{1}x}e^{-ik_{3}z}dxdz|$ $(c)$\\

$+{1\over |k_{2}|}|\int_{-r}^{r}\int_{|x|>s}\int_{|y|>t}{\partial f\over \partial y}(x,y,z)e^{-ik_{1}x}e^{-ik_{2}y}dxdy)e^{-ik_{3}z}dz|$ $(d)$\\

For $(a),(b)$, by the definition of normality, the fibres $\{f_{s,t},f_{-s,t},f_{s,-t},f_{-s,-t}\}$ are analytic at infinity, and of very moderate decrease, and, similarly, the fibres $\{{\partial f\over \partial x}_{s,t},{\partial f\over \partial x}_{-s,t},{\partial f\over \partial x}_{s,-t},{\partial f\over \partial x}_{-s,-t}\}$ are analytic at infinity, and of moderate decrease. Using the bound $val$ on the zeros, uniform in $(s,t)$, we can then repeat the calculation above $8$ times, to obtain that, uniformly in $r$;\\

$(a)+(b)\leq  {4\over |k_{1}k_{2}k_{3}|}(val+1)4\pi max(||f_{s,t}||,||f_{-s,t}||,||f_{s,-t}||,||f_{-s,-t}||)$\\

$+{4\over |k_{1}^{2}k_{2}k_{3}|}(val+1)4\pi max(||{\partial f\over \partial x}_{s,t}||,||{\partial f\over \partial x}_{-s,t}||,||{\partial f\over \partial x}_{s,-t}||,||{\partial f\over \partial x}_{-s,-t}||)$\\

$\leq {4C\over |k_{1}k_{2}k_{3}||(s,t)|}(val+1)4\pi+{4\over |k_{1}^{2}k_{2}k_{3}|(s,t)|^{2}}(val+1)4\pi$\\

For $(c)$, we have that ${\partial^{2}f\over \partial x^{2}}$ is of moderate decrease $3$, hence belongs to $L^{1}(\mathcal{R}^{2})$, so we can repeat the calculation above, to obtain that, uniformly in $r$;\\

$(c)\leq {1\over |k_{1}^{2}k_{2}|}|\int_{-\infty}^{\infty}\int_{|x|>s}|-{\partial^{2}f\over \partial x^{2}}(x,t,z)+{\partial^{2}f\over \partial x^{2}}(x,-t,z)|dxdz$\\

$\leq {4C\over |k_{1}^{2}k_{2}|}\int_{|x|>s}\int_{-\infty}^{\infty}{1\over (x^{2}+z^{2}+t^{2})^{3\over 2}}dzdx$\\

$\leq {4C\pi\over |k_{1}^{2}k_{2}|}\int_{|x|>s}{1\over (x^{2}+t^{2})}dx$\\

$\leq {4C\pi\over |k_{1}^{2}k_{2}|}\int_{|x|>s}{1\over x^{2}}dx$\\

$\leq {8C\pi\over |k_{1}^{2}k_{2}|s}$\\

We also have that;\\

$|\int_{|x|>s}{1\over (x^{2}+t^{2})}dx\leq {1\over t^{2}}\int_{-\infty}^{\infty}{1\over 1+{x^{2}\over t^{2}}}dx$\\

$={1\over t}[tan^{-1}({x\over t})]^{\infty}_{-\infty}$\\

$={\pi\over t}$\\

so that;\\

$(c)\leq min({8C\pi\over |k_{1}^{2}k_{2}|s},{4C\pi^{2}\over |k_{1}^{2}k_{2}|t})$\\

$\leq {4\sqrt{2}C\pi^{2}\over |k_{1}^{2}k_{2}||s,t|}$\\

For $(d)$, we can combine $(a),(b),(c)$ to obtain that;\\

$|\int_{-r}^{r}(\int_{|x|>s}\int_{|y|>t}f(x,y,z)e^{-ik_{1}x}e^{-ik_{2}y}dxdy)e^{-ik_{3}z}dz|$\\

$\leq {4C\over |k_{1}k_{2}k_{3}||(s,t)|}(val+1)4\pi+{4\over |k_{1}^{2}k_{2}k_{3}|(s,t)|^{2}}(val+1)4\pi+{4\sqrt{2}C\pi^{2}\over |k_{1}^{2}k_{2}||s,t|}$\\

$+{1\over |k_{2}|}|\int_{-r}^{r}\int_{|x|>s}\int_{|y|>t}{\partial f\over \partial y}(x,y,z)e^{-ik_{1}x}e^{-ik_{2}y}dxdy)e^{-ik_{3}z}dz|$\\

$\leq {E\over |(s,t)|}+{1\over |k_{2}|}|\int_{-r}^{r}\int_{|x|>s}\int_{|y|>t}{\partial f\over \partial y}(x,y,z)e^{-ik_{1}x}e^{-ik_{2}y}dxdy)e^{-ik_{3}z}dz|$\\

$|(s,t)|>1$\\

so that repeating the above argument with ${\partial f\over \partial y}$ replacing $f$;\\

$(d)\leq {F\over |(s,t)||k_{2}|}+{1\over |k_{2}|^{2}}|\int_{-r}^{r}\int_{|x|>s}\int_{|y|>t}{\partial^{2} f\over \partial y^{2}}(x,y,z)e^{-ik_{1}x}e^{-ik_{2}y}dxdy)e^{-ik_{3}z}dz|$\\

and, nesting the arguments, uniformly in $r$;\\

$|\int_{-r}^{r}(\int_{|x|>s}\int_{|y|>t}f(x,y,z)e^{-ik_{1}x}e^{-ik_{2}y}dxdy)e^{-ik_{3}z}dz|$\\

$\leq {E\over |(s,t)|}+{F\over |(s,t)||k_{2}|}+{G\over |s,t||k_{2}|^{2}}+{1\over |k_{2}|^{3}}|\int_{-r}^{r}\int_{|x|>s}\int_{|y|>t}{\partial^{3} f\over \partial y^{3}}(x,y,z)e^{-ik_{1}x}e^{-ik_{2}y}dxdy)e^{-ik_{3}z}dz|$\\

$(\dag\dag)$\\

Now, we can use the fact that ${\partial^{3}f\over \partial y^{3}}$ is of moderate decrease $4$, to see that ${\partial^{3}f\over \partial y^{3}}\in L^{1}(\mathcal{R}^{3})$, so that, uniformly in $r$, repeating the argument $(\dag)$;\\

$|\int_{-r}^{r}\int_{|x|>s}\int_{|y|>t}{\partial^{3} f\over \partial y^{3}}(x,y,z)e^{-ik_{1}x}e^{-ik_{2}y}dxdy)e^{-ik_{3}z}dz|$\\

$\leq \int_{-\infty}^{\infty}\int_{|x|>s}\int_{|y|>t}|{\partial^{3} f\over \partial y^{3}}|(x,y,z)dxdydz$\\

$\leq min(\int_{|x|>s}\int_{-\infty}^{\infty}\int_{-\infty}^{\infty}|{\partial^{3} f\over \partial y^{3}}|(x,y,z)dydzdx, \int_{|y|>t}\int_{-\infty}^{\infty}\int_{-\infty}^{\infty}|{\partial^{3} f\over \partial y^{3}}|(x,y,z)dxdzdy)$\\

$\leq min({2C\pi^{2}\over s},{2C\pi^{2}\over t})$\\

$\leq {2\sqrt{2}C\pi^{2}\over |(s,t)|}$\\

Now, from $(\dag\dag)$, we obtain that, uniformly in $r$;\\

$|\int_{-r}^{r}(\int_{|x|>s}\int_{|y|>t}f(x,y,z)e^{-ik_{1}x}e^{-ik_{2}y}dxdy)e^{-ik_{3}z}dz|$\\

$\leq {H\over |(s,t)|}$\\

which is $(iii)$. Combining with $(i),(ii)$, we obtain that, uniformly in $r$;\\\\

$|\int_{-r}^{r}(F(k_{1},k_{2},z)-\int_{-s}^{s}\int_{-t}^{t}f(x,y,z)e^{-ik_{1}x}e^{-ik_{2}y}dxdy)e^{-ik_{3}z}dz|$\\

$\leq {A_{k_{1}k_{2}k_{3}}\over s}+{B_{k_{1}k_{2}k_{3}}\over t}+{C_{k_{1}k_{2}k_{3}}\over |(s,t)|}$\\

where the constants $\{A_{k_{1}k_{2}k_{3}},B_{k_{1}k_{2}k_{3}},C_{k_{1}k_{2}k_{3}}\}\subset\mathcal{R}_{>0}$ can be read from the proof.\\

Applying the Moore-Osgood Theorem, it is then clear that;\\

$F(k_{1},k_{2},k_{3})=C(k_{1},k_{2},k_{3})$\\

$=lim_{r\rightarrow\infty,s\rightarrow\infty,t\rightarrow\infty}\int_{-r}^{r}\int_{-s}^{s}\int_{-t}^{t}f(x,y,z)e^{-ik_{1}x}e^{-ik_{2}y}e^{-ik_{3}z}dxdydz$\\

and similarly;\\

$G(k_{1},k_{2},k_{3})=B(k_{1},k_{2},k_{3})$\\

$=lim_{r\rightarrow\infty,s\rightarrow\infty,t\rightarrow\infty}\int_{-r}^{r}\int_{-s}^{s}\int_{-t}^{t}f(x,y,z)e^{-ik_{1}x}e^{-ik_{2}y}e^{-ik_{3}z}dxdydz$\\

$H(k_{1},k_{2},k_{3})=A(k_{1},k_{2},k_{3})$\\

$=lim_{r\rightarrow\infty,s\rightarrow\infty,t\rightarrow\infty}\int_{-r}^{r}\int_{-s}^{s}\int_{-t}^{t}f(x,y,z)e^{-ik_{1}x}e^{-ik_{2}y}e^{-ik_{3}z}dxdydz$\\

\end{proof}
\begin{lemma}
\label{normal5quasi}
If $f:\mathcal{R}^{3}\rightarrow\mathcal{R}$, the same result as Lemma \ref{normal5} holds, with the assumption that $f$ is quasi normal or quasi split normal.\\

\end{lemma}

\begin{proof}
Again, we can replace the use of $(i)-(vi)$ in Definition \ref{normal3}, within the proof of Lemma \ref{normal5}, by $(i)'-(vi)'$ or $(i)''-(vi)''$, with the argument used in Lemma \ref{normal4quasi}.
\end{proof}

\end{document}